\documentclass[journal, onecolumn, 12pt]{IEEEtran}

\usepackage{silence}
\usepackage{listings}
\usepackage{lipsum}

\usepackage[T1]{fontenc}
\usepackage[utf8]{inputenc}
\usepackage{amsmath,amssymb,amsfonts,mathrsfs,bm}
\usepackage{mathtools}
\usepackage{amsthm}
\usepackage{nicefrac}
\usepackage[shortlabels]{enumitem}
\usepackage{graphicx}
\usepackage{epstopdf}
\DeclareGraphicsExtensions{.eps,.png,.jpg,.pdf}

\usepackage{url}
\usepackage{colortbl}
\usepackage{booktabs}
\usepackage{multirow}
\usepackage[table,dvipsnames]{xcolor}
\usepackage[normalem]{ulem}
\usepackage{xparse}
\usepackage{calc}

\makeatletter
\@ifpackageloaded{natbib}{
	\relax
}{
	\usepackage{cite}
}
\makeatother

\usepackage{array}
\newcolumntype{L}[1]{>{\raggedright\let\newline\\\arraybackslash\hspace{0pt}}m{#1}}
\newcolumntype{C}[1]{>{\centering\let\newline\\\arraybackslash\hspace{0pt}}m{#1}}
\newcolumntype{R}[1]{>{\raggedleft\let\newline\\\arraybackslash\hspace{0pt}}m{#1}}

\makeatletter
\let\MYcaption\@makecaption
\makeatother
\usepackage[font=footnotesize]{subcaption}
\makeatletter
\let\@makecaption\MYcaption
\makeatother

\usepackage{glossaries}
\makeatletter
\let\oldgls\gls
\let\oldglspl\glspl

\newcommand\fussy@ifnextchar[3]{%
	\let\reserved@d=#1%
	\def\reserved@a{#2}%
	\def\reserved@b{#3}%
	\futurelet\@let@token\fussy@ifnch}
\def\fussy@ifnch{%
	\ifx\@let@token\reserved@d
		\let\reserved@c\reserved@a
	\else
		\let\reserved@c\reserved@b
	\fi
	\reserved@c}

\renewcommand{\gls}[1]{%
\oldgls{#1}\fussy@ifnextchar.{\@checkperiod}{\@}}
\renewcommand{\glspl}[1]{%
\oldglspl{#1}\fussy@ifnextchar.{\@checkperiod}{\@}}

\newcommand{\@checkperiod}[1]{%
	\ifnum\sfcode`\.=\spacefactor\else#1\fi
}
\makeatother

\newacronym{wrt}{w.r.t.}{with respect to}
\newacronym{RHS}{R.H.S.}{right-hand side}
\newacronym{LHS}{L.H.S.}{left-hand side}
\newacronym{iid}{i.i.d.}{independent and identically distributed}

\usepackage{float}

\ifx\notloadhyperref\undefined
	\ifx\loadbibentry\undefined
		\usepackage[hidelinks,hypertexnames=false]{hyperref}
	\else
		\usepackage{bibentry}
		\makeatletter\let\saved@bibitem\@bibitem\makeatother
		\usepackage[hidelinks,hypertexnames=false]{hyperref}
		\makeatletter\let\@bibitem\saved@bibitem\makeatother
	\fi
\else
	\ifx\loadbibentry\undefined
		\relax
	\else
		\usepackage{bibentry}
	\fi
\fi

\usepackage[capitalize]{cleveref}
\crefname{equation}{}{}
\Crefname{equation}{}{}
\crefname{claim}{claim}{claims}
\crefname{step}{step}{steps}
\crefname{line}{line}{lines}
\crefname{condition}{condition}{conditions}
\crefname{dmath}{}{}
\crefname{dseries}{}{}
\crefname{dgroup}{}{}

\crefname{Problem}{Problem}{Problems}
\crefformat{Problem}{Problem~(#2#1#3)}
\crefrangeformat{Problem}{Problems~(#3#1#4) to~(#5#2#6)}

\crefname{Theorem}{Theorem}{Theorems}
\crefname{Corollary}{Corollary}{Corollaries}
\crefname{Proposition}{Proposition}{Propositions}
\crefname{Lemma}{Lemma}{Lemmas}
\crefname{Definition}{Definition}{Definitions}
\crefname{Example}{Example}{Examples}
\crefname{Assumption}{Assumption}{Assumptions}
\crefname{Remark}{Remark}{Remarks}
\crefname{Rem}{Remark}{Remarks}
\crefname{remarks}{Remarks}{Remarks}
\crefname{Appendix}{Appendix}{Appendices}
\crefname{Supplement}{Supplement}{Supplements}
\crefname{Exercise}{Exercise}{Exercises}
\crefname{Theorem_A}{Theorem}{Theorems}
\crefname{Corollary_A}{Corollary}{Corollaries}
\crefname{Proposition_A}{Proposition}{Propositions}
\crefname{Lemma_A}{Lemma}{Lemmas}
\crefname{Definition_A}{Definition}{Definitions}

\usepackage{crossreftools}
\ifx\notloadhyperref\undefined
\else
	\relax
\fi

\usepackage{algorithm,algorithmic}

\ifx\loadbreqn\undefined
	\relax
\else
	\usepackage{breqn}
\fi

\newtheoremstyle{mystyle}{0pt}{0pt}{\upshape}{1em}{\itshape}{:}{ }{}
\theoremstyle{mystyle}

\ifx\renewtheorem\undefined
	\ifx\useTheoremCounter\undefined
		\newtheorem{Theorem}{Theorem}
		\newtheorem{Corollary}{Corollary}
		\newtheorem{Proposition}{Proposition}
		\newtheorem{Lemma}{Lemma}
	\else
		\newtheorem{Theorem}{Theorem}
		\newtheorem{Corollary}[Theorem]{Corollary}
		\newtheorem{Proposition}[Theorem]{Proposition}
	\fi

	\newtheorem{Remark}{Remark}

\fi

\theoremstyle{remark}

\theoremstyle{plain}

\newcommand{\qednew}{\nobreak \ifvmode \relax \else
		\ifdim\lastskip<1.5em \hskip-\lastskip
			\hskip1.5em plus0em minus0.5em \fi \nobreak
		\vrule height0.75em width0.5em depth0.25em\fi}



\NewDocumentCommand{\movedownsub}{e{^_}}{%
	\IfNoValueTF{#1}{%
		\IfNoValueF{#2}{^{}}
	}{%
		^{#1}
	}%
	\IfNoValueF{#2}{_{#2}}
}

\let\latexchi\chi
\RenewDocumentCommand{\chi}{}{\latexchi\movedownsub}

\newcommand{\ba}{\mathbf{a}}
\newcommand{\bA}{\mathbf{A}}
\newcommand{\bb}{\mathbf{b}}
\newcommand{\bB}{\mathbf{B}}
\newcommand{\bc}{\mathbf{c}}
\newcommand{\bC}{\mathbf{C}}

\newcommand{\bD}{\mathbf{D}}
\newcommand{\be}{\mathbf{e}}
\newcommand{\bE}{\mathbf{E}}
\newcommand{\boldf}{\mathbf{f}}
\newcommand{\bF}{\mathbf{F}}

\newcommand{\bh}{\mathbf{h}}
\newcommand{\bH}{\mathbf{H}}

\newcommand{\bI}{\mathbf{I}}

\newcommand{\bn}{\mathbf{n}}
\newcommand{\bN}{\mathbf{N}}

\newcommand{\bQ}{\mathbf{Q}}

\newcommand{\bR}{\mathbf{R}}

\newcommand{\bS}{\mathbf{S}}

\newcommand{\bu}{\mathbf{u}}

\newcommand{\bV}{\mathbf{V}}

\newcommand{\bW}{\mathbf{W}}

\newcommand{\bX}{\mathbf{X}}
\newcommand{\by}{\mathbf{y}}
\newcommand{\bY}{\mathbf{Y}}

\newcommand{\bZ}{\mathbf{Z}}

\DeclareSymbolFont{bsfletters}{OT1}{cmss}{bx}{n}
\DeclareSymbolFont{ssfletters}{OT1}{cmss}{m}{n}
\DeclareMathSymbol{\bsfGamma}{0}{bsfletters}{'000}
\DeclareMathSymbol{\ssfGamma}{0}{ssfletters}{'000}
\DeclareMathSymbol{\bsfDelta}{0}{bsfletters}{'001}
\DeclareMathSymbol{\ssfDelta}{0}{ssfletters}{'001}
\DeclareMathSymbol{\bsfTheta}{0}{bsfletters}{'002}
\DeclareMathSymbol{\ssfTheta}{0}{ssfletters}{'002}
\DeclareMathSymbol{\bsfLambda}{0}{bsfletters}{'003}
\DeclareMathSymbol{\ssfLambda}{0}{ssfletters}{'003}
\DeclareMathSymbol{\bsfXi}{0}{bsfletters}{'004}
\DeclareMathSymbol{\ssfXi}{0}{ssfletters}{'004}
\DeclareMathSymbol{\bsfPi}{0}{bsfletters}{'005}
\DeclareMathSymbol{\ssfPi}{0}{ssfletters}{'005}
\DeclareMathSymbol{\bsfSigma}{0}{bsfletters}{'006}
\DeclareMathSymbol{\ssfSigma}{0}{ssfletters}{'006}
\DeclareMathSymbol{\bsfUpsilon}{0}{bsfletters}{'007}
\DeclareMathSymbol{\ssfUpsilon}{0}{ssfletters}{'007}
\DeclareMathSymbol{\bsfPhi}{0}{bsfletters}{'010}
\DeclareMathSymbol{\ssfPhi}{0}{ssfletters}{'010}
\DeclareMathSymbol{\bsfPsi}{0}{bsfletters}{'011}
\DeclareMathSymbol{\ssfPsi}{0}{ssfletters}{'011}
\DeclareMathSymbol{\bsfOmega}{0}{bsfletters}{'012}
\DeclareMathSymbol{\ssfOmega}{0}{ssfletters}{'012}

\newcommand{\bLambda}{\bm{\Lambda}}

\newcommand{\bPhi}{\bm{\Phi}}

\makeatletter
\newcommand*\rel@kern[1]{\kern#1\dimexpr\macc@kerna}
\newcommand*\widebar[1]{%
  \begingroup
  \def\mathaccent##1##2{%
    \rel@kern{0.8}%
    \overline{\rel@kern{-0.8}\macc@nucleus\rel@kern{0.2}}%
    \rel@kern{-0.2}%
  }%
  \macc@depth\@ne
  \let\math@bgroup\@empty \let\math@egroup\macc@set@skewchar
  \mathsurround\z@ \frozen@everymath{\mathgroup\macc@group\relax}%
  \macc@set@skewchar\relax
  \let\mathaccentV\macc@nested@a
  \macc@nested@a\relax111{#1}%
  \endgroup
}
\makeatother

\DeclareMathOperator*{\argmax}{arg\,max}
\DeclareMathOperator*{\argmin}{arg\,min}

\DeclareMathOperator{\diag}{diag}

\DeclareMathOperator{\tr}{tr}

\DeclareMathOperator{\vect}{vec}

\DeclarePairedDelimiterX\ip[2]{\langle}{\rangle}{#1,#2}
\DeclarePairedDelimiterX\norm[1]{\lVert}{\rVert}{#1}
\DeclarePairedDelimiterXPP\col[1]{\operatorname{col}}{\{}{\}}{}{#1} 
\DeclarePairedDelimiterXPP\row[1]{\operatorname{row}}{\{}{\}}{}{#1} 
\DeclarePairedDelimiterXPP\erf[1]{\operatorname{erf}}{(}{)}{}{#1}
\DeclarePairedDelimiterXPP\erfc[1]{\operatorname{erfc}}{(}{)}{}{#1}
\DeclarePairedDelimiterXPP\op[2]{\operatorname{#1}}{(}{)}{}{#2} 

\newcommand{\bone}{\bm{1}}

\DeclarePairedDelimiterX\Set[2]\{\}{%

#2
}
\DeclarePairedDelimiterX\Setc[1]\{\}{%

#1
}

\NewDocumentCommand\set{s o m}{%
	\IfBooleanTF#1%
	{\IfValueTF{#2}{\Set*{#2}{#3}}{\Setc*{#3}}}%
	{\IfValueTF{#2}{\Set{#2}{#3}}{\Setc{#3}}}%
}

\NewDocumentCommand{\evalat}{s O{\big} m m}{%
\IfBooleanTF{#1}%
{{\left. #3 \right|_{#4}}}
{{#3#2|_{#4}}}%
}

\NewDocumentCommand \ifcond {m m} {%
	{#1} %
	\IfValueT{#2}{\, \middle|\, {#2}}%
}

\DeclareDocumentCommand \P {e{_} g >{\SplitArgument{ 1 }{ @| }}d() g } {%
	\mathbb{P}%
	\IfValueTF{#1}{_{#1}}
	{\IfValueT{#2}{_{#2}}}%
	\IfValueT{#3}{\left(\ifcond#3}%
	\IfValueT{#4}{\, \middle|\, {#4}}%
	\IfValueT{#3}{\right)}%
}

\DeclareDocumentCommand \E {e{_} g >{\SplitArgument{ 1 }{ @| }}o g } {%
\mathbb{E}%
\IfValueTF{#1}{_{#1}}
{\IfValueT{#2}{_{#2}}}%
\IfValueT{#3}{\left[\ifcond#3}%
\IfValueT{#4}{\, \middle|\, {#4}}%
\IfValueT{#3}{\right]}%
}

\ExplSyntaxOn
\NewDocumentCommand \dist {m o o} {%
\mathrm{#1}\left(%
	\IfValueT{#3}{%
		\tl_if_blank:nTF{ #3 }{\cdot\, \middle|\, }{#3\, \middle|\, }%
	}
	\IfValueT{#2}{#2}%
\right)%
}
\ExplSyntaxOff

\NewDocumentCommand {\cbrace} { D[]{black} d[] D(){\widthof{#5}} m m } {%
	\begingroup%
		\color{#1}
		\IfValueTF{#2}{%
			\overbrace{{\color{#1}#4}}^%
		}{
			\underbrace{#4}_%
		}%
		{\parbox[c]{#3}{\centering\footnotesize{#5}}}%
	\endgroup%
}

\let\oldforall\forall
\renewcommand{\forall}{\oldforall \, }

\let\oldexist\exists
\renewcommand{\exists}{\oldexist \, }

\graphicspath{{./Figures/}{./figures/}}
\newcommand{\includeCroppedPdf}[2][]{%
	\IfFileExists{./Figures/#2-crop.pdf}{}{%
		\immediate\write18{pdfcrop ./Figures/#2 ./Figures/#2-crop.pdf}}%
	\includegraphics[#1]{./Figures/#2-crop.pdf}}

\definecolor{gray90}{gray}{0.9}

\ifx\nohighlights\undefined

	\newcommand{\msout}[1]{\text{\color{green} \sout{\ensuremath{#1}}}}
	\newcommand{\del}[1]{{\color{green}\ifmmode \msout{#1}\else\sout{#1}\fi}}
\else

	\newcommand{\msout}[1]{#1}
	\newcommand{\del}[1]{#1}
\fi

\newcommand{\hhide}[1]{}

\ifx\diagnoselabel\undefined
	\relax
\else
	\makeatletter
	\def\@testdef #1#2#3{%
		\def\reserved@a{#3}\expandafter \ifx \csname #1@#2\endcsname
			\reserved@a  \else
			\typeout{^^Jlabel #2 changed:^^J%
				\meaning\reserved@a^^J%
				\expandafter\meaning\csname #1@#2\endcsname^^J}%
			\@tempswatrue \fi}
	\makeatother
\fi

\usepackage{cases}

  \def\R{{\mathbb{R}}} \def\C{{\mathbb{C}}}   \def\E{{\mathbb{E}}}
\def\bone{\mathbf{1}}

\newcommand{\beq}{\begin{eqnarray}}
\newcommand{\eeq}{\end{eqnarray}}

\def\cA{{\mathcal{A}}}  \def\cC{{\mathcal{C}}} \def\cD{{\mathcal{D}}}

 \def\cN{{\mathcal{N}}} \def\cO{{\mathcal{O}}} 
 \def\cR{{\mathcal{R}}}  
 \def\cV{{\mathcal{V}}}

           \def\lA{\left\|}     \def\rA{\right\|}

\makeatletter
\renewenvironment{proof}[1][\proofname]{\par
  \pushQED{\qed}%
  \normalfont \topsep0\p@\relax
  \trivlist
  \item[\hskip3\labelsep\itshape#1\@addpunct{:}]\ignorespaces}{%
  \popQED\endtrivlist\@endpefalse
}
\makeatother

\newacronym{MLE}{MLE}{maximum likelihood estimate}

\usepackage{setspace}
\begin{document}

\pagenumbering{arabic}

\title{MMV-Based Sequential AoA and AoD Estimation for Millimeter Wave MIMO Channels}
\author{Wei Zhang, Miaomiao Dong, and Taejoon Kim
\thanks{

{W. Zhang is with the School of Electrical and Electronic Engineering, Nanyang Technological University, Singapore (e-mail: weizhang@ntu.edu.sg).}
{M. Dong was with the Department of Electrical Engineering, City University of Hong Kong and the Department of Electrical Engineering and Computer Science, University of Kansas, KS 66045, USA (email: miao4600@163.com).}
{T. Kim is with the Department of Electrical Engineering and Computer Science, University of Kansas, KS 66045, USA (e-mail: taejoonkim@ku.edu).}
}
}
\maketitle
\vspace{-0.4cm}
\begin{abstract}
The fact that the millimeter-wave (mmWave) multiple-input multiple-output (MIMO) channel has sparse support in the spatial domain has motivated recent compressed sensing (CS)-based mmWave channel estimation methods, where the angles of arrivals (AoAs) and angles of departures (AoDs) are quantized using angle dictionary matrices. However, the existing CS-based methods usually obtain the estimation result through one-stage channel sounding that have two limitations: (i) the requirement of large-dimensional dictionary and (ii) unresolvable quantization error. These two drawbacks are irreconcilable; improvement of the one implies deterioration of the other. To address these challenges, we propose, in this paper, a two-stage method to estimate the AoAs and AoDs of mmWave channels. In the proposed method, the channel estimation task is divided into two stages, Stage I and Stage II. Specifically, in Stage I, the AoAs are estimated by solving a multiple measurement vectors (MMV) problem. In Stage II, based on the estimated AoAs, the receive sounders are designed to estimate AoDs. The dimension of the angle dictionary in each stage can be reduced, which in turn reduces the computational complexity substantially. We then analyze the successful recovery probability (SRP) of the proposed method, revealing the superiority of the proposed framework over the existing one-stage CS-based methods. We further enhance the reconstruction performance by performing resource allocation between the two stages. We also overcome the unresolvable quantization error issue present in the prior techniques by applying the atomic norm minimization method to each stage of the proposed two-stage approach. The simulation results illustrate the substantially improved performance with low complexity of the proposed two-stage method.

\end{abstract}
\vspace{-0.3cm}
\begin{IEEEkeywords}
Millimeter wave communications, compressed sensing, channel estimation, multiple-input multiple-output system, support recovery, and sequential estimation.
\end{IEEEkeywords}

\section{Introduction}

The spectrum-rich millimeter-wave (mmWave) frequencies between $30-300$ GHz have the potential to alleviate the current spectrum crunch in sub-6GHz bands that service providers are already experiencing. This major potential of the mmWave band has made it one of the most important components of future mobile cellular and emerging WiFi networks.
However, due to significant differences between systems
operating in mmWave and legacy sub-6 GHz bands, providing reliable and low-delay communication in
the mmWave bands is extremely challenging. Specifically,
to achieve the high spectral efficiency of mmWave communications, accurate channel state information
(CSI) is the key \cite{spatially,Heath16, Hur13,ZhangSD,ZhangSequ}, which is, however, challenging due to the high dimensionality of the channel as well as the mmWave hardware constraints.

Nevertheless, the mmWave  multiple-input multiple-output (MIMO) channel exhibits sparse property \cite{5gWhite,hur2016proposal}, facilitating the sparse channel representation by using small numbers of the angles of arrivals
(AoAs), angles of departures (AoDs), and path gains.
Typically,  by approximating the AoAs and AoDs to be on quantized angle grids, the compressed sensing (CS)-based approaches transform the AoA and AoD estimation problem to a sparse signal recovery problem \cite{alk,OMPchannel},
where the transmitter sends the channel sounding beams to the receiver and the receiver jointly estimates AoAs and AoDs. We refer to this method as the one-stage channel sounding scheme.
In particular, due to easy implementation and amenability for analysis, the orthogonal matching pursuit (OMP) has been widely studied \cite{OMPchannel,MultiOMP,cstOMP,jOMP,duan}.
The OMP iteratively searches a pair of AoA and AoD over an over-complete dictionary.
However, the computational complexity of OMP  increases quadratically with the sizes of the dictionaries, i.e., $O(L K G_r G_t)$, where $K$ is the number of channel uses for the channel sounding,
$L$ is the number of channel paths, and $G_r$ and $G_t$ are the dimensions of angle dictionaries for AoA and AoD, respectively.
It is worth pointing out that when the dimensions of the over-complete dictionaries, i.e., $G_r$ and $G_t$,
 increase, the complexity of the one-stage CS-based methods such as OMP becomes exceedingly impractical.

The over-complete dictionary and high computational complexity issues have been addressed in an adaptive-CS point-of-view with the primary focus on the sensing vector adaptation to the previous observations \cite{Hur13,alk,twoStageC}. Theoretically, it has been shown that the adaptive CS can be benificial in low SNR \cite{OptimalAdaptive}.
The multi-level (hierarchical) AoA and AoD search techniques \cite{Hur13,alk} leveraged the feedback, where the receiver conveys a feedback to the transmitter to guide the next level angle dictionary design. It is worth noting that these adaptation methods \cite{Hur13,alk} need multiple feedbacks and its performance critically relies on the reliability of the feedback.
To reduce the feedback overhead, a two-stage CS was proposed in \cite{twoStageC}, where the first stage is to obtain a coarse estimation of the support set and the second stage refines the result of the first stage. This method \cite{twoStageC} only requires one-time feedback, but achieves compatible estimation performance in low SNR.

\subsection{Our Contributions}
We newly study a sequential, two-stage AoA and AoD estimation framework for reduced computational complexity and improved estimation performance.
Specifically, in Stage I, the support set of AoAs is recovered  at the receiver by solving a multiple measurement vectors (MMV) problem.
Leveraging the shared sparse set, it has been found that the MMV approach can provide improved estimation performance compared to the single measurement vector (SMV) approach \cite{ChenMMV,vanMMV,LeeMMV}.
In Stage II, the receiver estimates the AoDs of the channel by exploiting the estimated AoAs from Stage I.
Importantly, the estimated AoAs guide the design of receive sounding signals, which saves the channel use overhead and improves the accuracy of AoD estimation.
In each stage, since we only estimate AoAs or AoDs,
the dimensions of the signal and angle dictionary
are much smaller than those of
the one-stage joint AoA and AoD estimation
\cite{OMPchannel,cstOMP,jOMP},
readily reducing the computational complexity substantially. This can be viewed as of converting the multiplicative channel sounding overhead (e.g., $\cO(G_r G_t)$ of OMP) to an additive overhead.

By analyzing the MMV statistics, we present a lower bound for the successful probability of recovering the support sets.
Furthermore, based on the successful recovery probability (SRP) analysis of the proposed two-stage method, a resource allocation (between Stage I and Stage II) strategy is newly proposed to improve SRPs for both AoA and AoD estimation.
The numerical results validate the efficacy of the proposed resource allocation method.

Finally, in order to address the issue of unresolvable quantization error, we extend the proposed two-stage method to the one with super resolution. Specifically, in each stage of AoA or AoD estimation, we reformulate the MMV problem as an atomic norm minimization problem \cite{OffGridCS,MmvAtomic,superMM}, which is solved by using alternating direction method of multipliers (ADMM). Compared to the dictionary-based methods, the atomic norm minimization can be thought of as the case when the infinite dictionary matrix is employed.
We demonstrate through simulations that the quantization error of the two-stage method with super resolution can be effectively reduced.

\subsection{Paper Organization and Notations}
The paper is organized as follows. In Section \ref{section model}, we introduce the signal model and the CS-based channel estimation problem. In Section \ref{section algorithm}, based on the angular-domain channel representation, the proposed sequential AoA and AoD estimation method is presented.
In Section \ref{sec analyze}, we analyze the proposed method in terms of SRP and introduce the resource allocation strategy. In Section \ref{section atomic}, the atomic norm-based design is described, which resolves the quantization error in the estimated AoAs and AoDs.
The simulation results and conclusion are presented in Section \ref{section simulation} and Section \ref{section conclusion}, respectively.

\emph{Notations:} A bold lower case letter $\mathbf{a}$ is a vector and a bold capital letter $\bA$ is a matrix. ${{\bA}^{T}}$, ${{\bA}^{*}}$, ${{\bA}^{H}}$, ${{\bA}^{-1}}$, $\tr(\bA)$, $\left| \bA \right|$,  ${{\left\| \bA \right\|}_{F}}$ and ${{\left\| \mathbf{a} \right\|}_{2}}$ are, respectively, the transpose, conjugate, Hermitian, inverse, trace, determinant, Frobenius norm of $\bA$, and $\ell_2$-norm of $\mathbf{a}$.
$\bA^{\dagger} = (\bA^H \bA)^{-1}\bA^H$ denotes the pseudo inverse of a tall matrix $\bA$.
${{[\bA]}_{:.i}}$, ${{[\bA]}_{i,:}}$, ${{[\bA]}_{i,j}}$, and $[\ba]_i$ are, respectively, the $i$th column, $i$th row, $i$th row and $j$th column entry of $\bA$, and $i$th entry of vector $\mathbf{a}$. $\mathrm{\mathop{vec}}(\bA)$ stacks the columns of $\bA$ and forms a long column vector.
$\mathrm{\mathop{diag}}(\mathbf{a})$ returns a square diagonal matrix with the vector $\ba$ on the main diagonal.
 ${{\mathbf{I}}_{M}}\in {{\mathbb{R}}^{M\times M}}$ is the $M$-dimensional identity matrix.
 {The $\mathbf{1}_{M,N}  \in  \R^{M\times N}$ and $\mathbf{0}_{M,N}  \in  \R^{M\times N}$ are the all one matrix, and zero matrix, respectively.}
 $\cR(\bF)$ denotes the subspace spanned by the columns of matrix $\bF$. $\bA\otimes \mathbf{B}$ and  $\bA \circ \mathbf{B}$ denote the Kronecker product and Khatri-Rao product of $\bA$ and $\mathbf{B}$, respectively. The $\lceil x \rceil$ returns the smallest integer greater than or equal to $x$.

\section{System Model and General Statement of Techniques} \label{section model}

\subsection{Channel Model}
The mmWave transmitter and receiver  are equipped with $N_t$ and $N_r$ antennas, respectively.
Suppose that the number of separable paths between the transmitter and receiver is $L$, where $L\ll \min\{ N_r, N_t \}$.
The physical mmWave channel representation based on the uniform linear array \cite{li2017millimeter,OMPchannel,wan2019compressive,zhang2020downlink} is given by\footnote{
In wideband communication systems, one can model the channel as constant AoA/AoD and varying path gains \cite{qin2018time,park2018spatial}. Here we could also assume a narrow band block fading channel where the channel is static during the channel coherence time. The CSI acquisition and data transfer are framed to happen within the channel coherence time \cite{li2017millimeter,OMPchannel,wan2019compressive,zhang2020downlink}.
},
\begin{align}
\bH=\sqrt{\frac{N_rN_t}{L}}\sum\limits_{l=1}^{L}{{{\alpha }_{l}}}\mathbf{a}_r({f_{r,l}}){\ba_t^H}({f_{t,l}}), \label{channel model}
\end{align}
where  $\mathbf{a}_t(\cdot)\in \C^{N_t \times 1}$ and $\mathbf{a}_r(\cdot) \in \C^{N_r \times 1}$  are the array response vectors of the transmit and receive antenna arrays. Specifically,  $\mathbf{a}_t(f)$ and  $\mathbf{a}_r(f)$ are given by
$\mathbf{a}_t(f)=\frac{1}{\sqrt{N_t}}{{\left[ 1,{{e}^{j2\pi f}},\ldots ,{{e}^{j2\pi (N_t-1)f}} \right]}^{T}} $
and
$\mathbf{a}_r(f)=\frac{1}{\sqrt{N_r}}{{\left[ 1,{{e}^{j2\pi f}},\ldots ,{{e}^{j2\pi (N_r-1)f}} \right]}^{T}}$, where $f\in [0,1)$ is the normalized spatial angle. Here we assume ${{f}_{r,l}} $ and ${{f}_{t,l}}$ in \eqref{channel model} are independent and uniformly distributed in $[0,1 )$, and the gain of the $l$th path $\alpha_l$ follows the complex Gaussian distribution, i.e., ${{\alpha }_{l}} \sim \mathcal{C}\mathcal{N}(0,\sigma_l^2)$.
Angular domain representation of the channel in \eqref{channel model} can be rewritten as
\begin{align}
\bH={{\bA}_{r}}\diag(\mathbf{h})\bA_{t}^H, \label{compact channel}
\end{align}
where ${{\bA}_{r}}=[\mathbf{a}_r({{f}_{r,1}}),\ldots ,\mathbf{a}_r({{f}_{r,L}})]\in {\bC^{N_r\times L}}$ , ${{\bA}_{t}}=[\mathbf{a}_t({{f}_{t,1}}),\ldots ,\mathbf{a}_t({{f}_{t,L}})]\in {\bC^{N_t\times L}}$, and $\bh =[h_1,\ldots,h_L]\in \C^{L \times 1}$ with $h_l=\sqrt{\frac{N_rN_t}{L}}{{\alpha }_{l}}$, $l=1,\ldots ,L$.

\begin{figure}
\centering
\includegraphics[width=.7\textwidth]{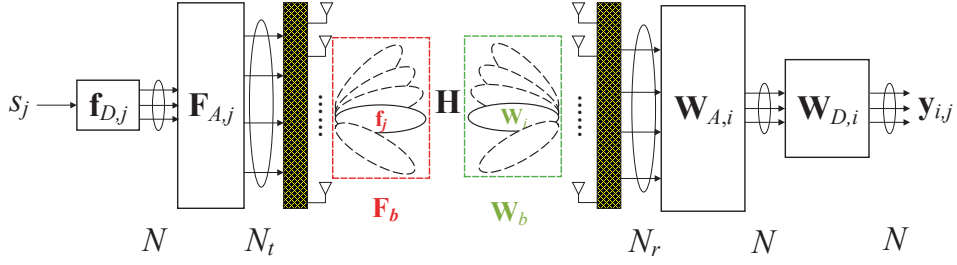}
\caption{Conventional one-stage mmWave channel sounding} \label{system diagram}
\end{figure}

\subsection{Channel Sounding} \label{section channel sounding}
Fig. \ref{system diagram} illustrates the conventional one-stage mmWave channel sounding operation, where the transmitter and receiver are equipped with the large-dimensional hybrid analog-digital MIMO arrays that are driven by a limited number of RF chains, i.e., $N\ll \min\{ N_t, N_r \}$.
In each channel use of downlink channel sounding,
the transmitter generates a beam conveying the pilot signal and the receiver simultaneously generates $N$ separate beams, using the $N$ RF chains, to obtain a $N$-dimensional observation.
We let the numbers of the transmit sounding beams (TSBs) and receive sounding beams (RSBs)  for channel estimation be $B_t$ and $B_r$, respectively.
For convenience, we assume that $B_r$ is an integer multiple of $N$.
The total number of channel uses for the conventional one-stage sounding process is then $K=B_rB_t/N$.
Specifically, the RSB matrix in Fig. \ref{system diagram} is given by
\begin{align}
\mathbf{W}_b=[{\bW_{1}},{\bW_{2}},\ldots ,{\bW_{B_r/N}}]\in {\bC^{N_r\times {B_r}}}, \label{RSB matrix}
\end{align}
where $\bW_i \in {\bC^{N_r\times N}}$ for $i=1,2,\ldots,B_r/N$,
and $\bW_i=\bW_{A,i}\bW_{D,i}$ with $\bW_{A,i}\in {\bC^{N_r\times N}}$ and $\bW_{D,i} \in {\bC^{N\times N}}$ being the receive analog and digital sounders, respectively.
Similarly, the TSB matrix is given by
\begin{align}
\bF_b=[{{\mathbf{f}}_{1}},{{\mathbf{f}}_{2}},\cdots ,{\mathbf{f}}_{B_t}]\in {\bC^{N_t\times {B_t}}}, \label{TSB matrix}
\end{align}
where  $\boldf_j\in {\bC^{N_t\times 1}}$ for $j=1,2,\cdots,B_t$ is the $j$th transmit sounder,
and $\boldf_j={\bF_{A,j}}{\boldf_{D,j}}s_j$  with ${\bF_{A,j}}\in {\bC^{N_t\times N}}$ and $\boldf_{D,j}\in {\bC^{N\times 1}}$ being the transmit analog and digital sounders, respectively.
 Each observation $\by_{i,j}\in \C^{N\times 1}$ in Fig. \ref{system diagram}, associated with the $i$th RSB and $j$th TSB, $i\in \{ 1,\ldots, B_r/N \}$ and $j \in \{1,2,\ldots, B_t \}$, can be expressed as
\begin{align} {\by}_{i,j}=\mathbf{W}_{i}^H\bH{{\mathbf{f}}_{j}}s_j+\mathbf{W}_{i}^H{{\mathbf{n}}_{j}}.
\label{single channel uses}
\end{align}
 The $s_j$ denotes the training signal and without loss of generality, we let $s_j=1$.
It is worth noting that only phase shifters are employed to constitute the analog arrays for power saving, where $|{{[{\bW_{A,i}}]}_{m,n}}|=1/\sqrt{N_r}$, and $|{{[{\bF_{A,j}}]}_{m,n}}|=1/\sqrt{N_t}, \forall m,n$.
Moreover,  the power constraint $\left\| {{\mathbf{f}}_{j}} \right\|_{2}^{2} =  p$ is imposed to the transmit sounding beam at each channel use with $p$ being the power budget, and the noise vector follows ${{\mathbf{n}}_{j}} \sim \mathcal{C}\mathcal{N}(\mathbf{0}_{N_r},{{\sigma }^{2}}{{\mathbf{I}}_{N_r}})$. Thus, the signal to noise ratio is $p/{{\sigma }^{2}}$.

We collect all observations in \eqref{single channel uses} by using $\bW_b$ in \eqref{RSB matrix} and $\bF_b$ in \eqref{TSB matrix} as
\begin{align}
\mathbf{Y}={\bW_b^H}\bH \bF_b+{\bW_b^H}\mathbf{N}, \label{matrix observations}
\end{align}
where $\bY \in \C^{B_r\times B_t}$ and $\bN =[ \bn_1, \ldots, \bn_{B_t} ] \in \C^{N_r \times B_t}$.
For example,
$\bW_b$ and $\bF_b$ in \eqref{matrix observations}
can be generated randomly \cite{ZhangSD} or designed as a partial discrete Fourier transform (DFT) matrix \cite{OMPchannel}.
We assume that the number of observations is strictly lower than the dimension of the channel matrix, i.e.,
${B_r}{B_t}\ll N_rN_t$.
The channel estimation task is to utilize the observations in \eqref{single channel uses} (equivalently, \eqref{matrix observations}) to obtain the estimate of the channel matrix $\bold{H}$ in \eqref{compact channel}.
Encountering \eqref{compact channel}, the channel estimation task boils down to reconstructing $\{{{f}_{r,1}},\ldots ,{{f}_{r,L}}\}$, $\{{{f}_{t,1}},\ldots ,{{f}_{t,L}}\}$ and $\{{{h}_{1}},\ldots ,{{h}_{L}}\}$ from the observations.

\subsubsection{Oracle Estimator}

The oracle estimator that we will utilize for benchmark\footnote{
Both Cramer-Rao lower bound (CRLB) \cite{Ahmed2020VTC} and the oracle estimator \cite{OMPchannel} can be utilized to evaluate the accuracy of estimation algorithms.
Since the CRLB can only be calculated for one-stage method, in this work we use the oracle estimator as the benchmark instead.} is obtained by assuming perfect knowledge of AoAs and AoDs in \eqref{compact channel}.
The oracle channel estimate only needs to estimate the path gain $\mathbf{h}$, thus the channel estimate is expressed  as $\widehat{\bH}= \bA_r \diag(\widehat{\bh}) \bA_t^H$, where $\diag(\widehat{\bh}) \in \C^{L \times 1}$ is the solution to the following problem:
\begin{align}
\widehat{\bh} = \argmin_{\bh} \| \bY - {\bW_b^H}\bA_r \diag({\bh}) \bA_t^H \bF_b\|_F^2. \label{oracal estimator}
\end{align}
Because \eqref{oracal estimator} is convex, the optimal solution is $\widehat{\bh} = (\bX^H\bX)^{-1} \bX^H \vect(\bY)$, where $\bX \in \C^{B_rB_t\times L}$ is given by
$\bX =  \left[\vect {([\bW_b^H \bA_r]_{:,1} [\bA_t^H \bF_b]_{1,:})}, \ldots,\vect{([\bW_b^H \bA_r]_{:,L} [\bA_t^H \bF_b]_{L,:})}   \right]$.
Because we have $B_r B_t \gg L$, $\bX^H \bX$ is invertible.

\subsection{Compressed Sensing-Based Channel Estimation}\label{traditional CS}

Recalling the channel model in \eqref{compact channel},
a typical CS framework restricts the normalized spatial angles ${{f}_{r,l}},{{f}_{t,l}},~l=1,2,\ldots, L$, to be chosen from the discrete angle dictionaries,
$
{{f}_{r,l}}\in \left[0,{1}/{G_r},\ldots, {(G_r-1)}/{G_r}\right]$, and ${{f}_{t,l}}\in \left[0,{1}/{G_t},\ldots, {(G_t-1)}/{G_t}\right] $,
where $G_r=\lceil sN_r \rceil$ and $G_t=\lceil s N_t \rceil$ with $s \ge 1$ are, respectively, the cardinalities of the receive and transmit spatial angle dictionaries.
The transmit and receive array response dictionaries are then given by
\begin{align}
~~~\bar{\bA}_r=\left[\mathbf{a}_r(0),\mathbf{a}_r\left(\frac{1}{G_r}\right),\ldots ,\mathbf{a}_r\left(\frac{G_r-1}{G_r}\right)\right]\in {\bC^{N_r\times G_r}}\nonumber
\end{align}
and
\begin{align}	
~~~\bar{\bA}_t=\left[\mathbf{a}_t(0),\mathbf{a}_t\left(\frac{1}{G_t}\right),\ldots ,\mathbf{a}_t\left(\frac{G_t-1}{G_t}\right)\right]\in {\bC^{N_t\times G_t}}. \nonumber
\end{align}
For the latter array response dictionaries, the channel model in \eqref{compact channel} can be rewritten as
\begin{align}
\bH= \bar{\bA}_r{\bar{\bH}_a}\bar{\bA}_t^H + \bE, 	\label{redundant channel estimation}
\end{align}
where ${\bar{\bH}_a}\in {\bC^{G_r\times G_t}}$ is an $L$-sparse matrix with $L$ non-zero entries corresponding to the positions of AoAs and AoDs on their respective  angle grids, and $\bE \in \C^{N_r \times N_t}$ denotes the quantization error.

 Because the dictionary matrices $\bar{\bA}_r$ and $\bar{\bA}_t$ are known, the channel estimation task is equivalent to estimating the non-zero entries in $\bar{\bH}_a$. Plugging the model in \eqref{redundant channel estimation} into \eqref{matrix observations} gives
\begin{align}
\mathbf{Y}={\bW_b^H}\bar{\bA}_r({\bar{\bH}_a}+\bE)\bar{\bA}_t^H\bF_b+{\bW_b^H}\mathbf{N}.  \label{bold Y}	\end{align}
Vectorizing $\bY$ in \eqref{bold Y} yields
\begin{align}
     \vect(\mathbf{Y})
 & =({\bF_b^{T}}\bar{\bA}_t^*\otimes {\bW_b^H}\bar{\bA}_r)(\vect({\bar{\bH}_a} + \bE))+\vect({\bW_b^H}\mathbf{N}).   \label{vector cs}
\end{align}
Denoting $\mathbf{D}={\bF_b^{T}}\bar{\bA}_t^*\otimes {\bW_b^H}\bar{\bA}_r\in {\bC^{{B_r}{B_t}\times G_r{{G}_{b}}}}$ and $\bar{\mathbf{n}}=\bD \vect(\bE)+\vect({\bW_b^H}\mathbf{N})\in {\bC^{{B_r}{B_t}\times 1}}$ gives
$
\vect(\mathbf{Y})=\mathbf{D}\vect({\bar{\bH}_a})+\bar{\mathbf{n}}$.
Hence, the estimation of $\vect({\bar{\bH}_a})$ from \eqref{vector cs} can be stated as a sparse signal reconstruction problem:
\begin{align}
 \min_{\bar{\bH}_a}  \| \vect(\mathbf{Y}) - \mathbf{D}\vect({\bar{\bH}_a}) \|_2 ~\text{subject to } \| \vect( {\bar{\bH}_a})  \|_0=L, \label{vector observations}
\end{align}
where $\|\cdot \|_0$ is the  $\ell_0$-norm that returns the number of non-zero coordinates of a vector. The problem in \eqref{vector observations} can be solved by using standard CS methods \cite{DonohoCS,TroppOMP}.

The number of required observations to reconstruct $L$-sparse vector $\vect(\bar\bH_a)\in \C^{G_r G_t \times 1} $ in \eqref{vector observations} has previously characterized as $O\left( L\cdot \log (G_rG_t) \right)$ \cite{DonohoCS}, which is much smaller than $O(N_r N_t)$.
However, the computational complexity for estimating $\vect(\bar{\bH}_a)$ in \eqref{vector observations} by using OMP, for example, is $O(LB_rB_t G_r G_t)$.
Though the quantization error associated with using dictionaries can be made small by increasing the sizes of the dictionaries, the growing computational complexity remains a critical challenge.
Instead of developing another one-stage channel sounding method (as in Fig. \ref{system diagram}), we propose a new two-stage channel sounding and estimation framework to overcome the large overhead and complexity drawbacks.

\begin{figure}
\centering
\includegraphics[width=.57\textwidth]{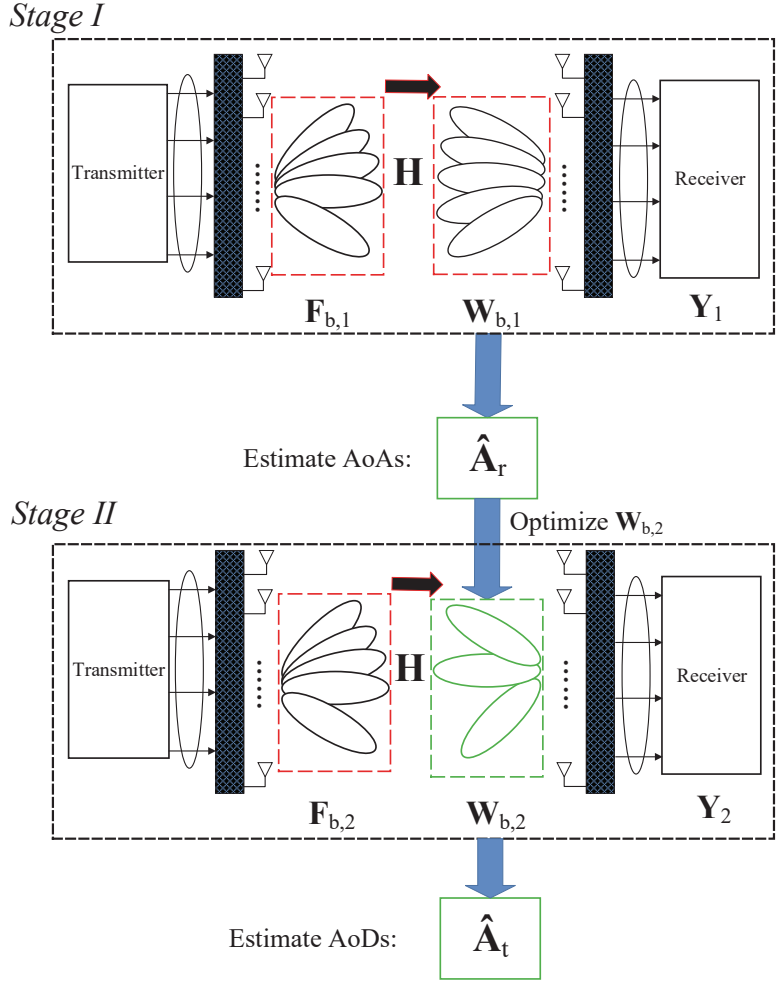}
\caption{Illustration of the proposed two-stage AoA and AoD estimation.} \label{alg diagram}
\end{figure}

\section{Two-stage AoA and AoD Estimation} \label{section algorithm}
A conceptual diagram of the proposed two-stage AoA and AoD estimation framework is presented in Fig. \ref{alg diagram}. The proposed sequential technique has constituent two stages of channel sounding, where each stage exclusively exploits much low-dimensional dictionary compared to the one-stage channel sounding in Fig. \ref{system diagram}.

Under the similar definitions of one-stage method in \eqref{matrix observations}, in Stage I of the two-stage framework of Fig. \ref{alg diagram}, the transmit and receive sounding beams are represented by ${\bF_{b,1}}\in {\bC^{N_t\times B_{t,1}}}$ and ${\bW_{b,1}}\in {\bC^{N_r\times B_{r,1}}}$, respectively.
 The AoA estimates of Stage I produce the estimation of array response matrix $\bA_r$ in \eqref{compact channel}, i.e., $\widehat{\bA}_r \in \C^{N_t \times L}$.
In Stage II, the transmit and receive sounding beams are denoted by ${\bF_{b,2}}\in {\bC^{N_t\times B_{t,2}}}$ and ${\bW_{b,2}}\in {\bC^{N_r\times B_{r,2}}}$, respectively.
In particular, the receive sounding beams ${\bW_{b,2}}$ is optimized based on the estimated AoA array response matrix $\widehat{\bA}_r$ at Stage I, which leads to improved estimation accuracy as our analysis and simulation show.
The total number of observations is given by
$N_p = B_{t,1}B_{r,1}+B_{t,2}B_{r,2}.$
Accordingly, the total number of channel uses is $K=(B_{t,1}B_{r,1}+B_{t,2}B_{r,2})/N$.

\subsection{Stage I: AoA Estimation}
We rewrite the channel model in \eqref{redundant channel estimation} as
  $ \bH
   =\bar{\bA}_r\bar{\bH}_a\bar{\bA}_t^H +\bE
 =\bar{\bA}_r{{\mathbf{Q}}_{r}}+\bE$,
where ${{\mathbf{Q}}_{r}}\in {\bC^{G_r\times N_t}}$ has $L$ non-zero rows, whose indices are collected into the support set $\Omega_r \subset \{1,2,\ldots,G_r \}$  and $|\Omega_r| = L$.
Using $\Omega_r$, the $\bA_r$ in \eqref{compact channel} can be written using the columns of $\bar{\bA}_r$ indexed by $\Omega_r$ as $[\bar{\bA}_r]_{:,\Omega_r} = \bA_r$. \par
To estimate the AoAs, we need to recover the support set $\Omega_r$.
Similar to the one-stage sounding in \eqref{matrix observations}, at Stage I in Fig.~\ref{alg diagram}, the observations ${{\mathbf{Y}}_{1}} \in \C^{B_{r,1}\times B_{t,1}}$ is expressed as
\begin{align}	
{{\mathbf{Y}}_{1}} &=\mathbf{W}_{b,1}^H \bH {\bF_{b,1}}+\mathbf{W}_{b,1}^H\mathbf{N}_1\nonumber\\
 &=\mathbf{W}_{b,1}^H\bar{\bA}_r{{\mathbf{Q}}_{r}}{\bF_{b,1}}+ \mathbf{W}_{b,1}^H\bE{\bF_{b,1}} +\mathbf{W}_{b,1}^H\mathbf{N}_1\nonumber\\
 &=\bPhi_1 \bC_1+ \mathbf{W}_{b,1}^H\bE{\bF_{b,1}} +\mathbf{W}_{b,1}^H\mathbf{N}_1,
\label{MMV observation}
\end{align}
where $\bPhi_1 =\bW_{b,1}^H\bar{\bA}_r  \in \C^{B_{r,1} \times G_r}$, $\bC_1 =  {{\mathbf{Q}}_{r}}{\bF_{b,1}} \in \C^{G_r \times B_{t,1}}$, and $\mathbf{N}_1\in {\bC^{N_r\times B_{t,1}}}$ is the noise matrix with \gls{iid} entries according to ${{[\mathbf{N}_1]}_{i,j}}\sim \mathcal{C}\mathcal{N}(0,{{\sigma }^{2}})$, $\forall i,j$.
Due to the row sparsity of $\bQ_r$,  it is clear that $\bC_1$ also has $L$ non-zero rows indexed by $\Omega_r$.
If $B_{t,1}=1$, the recovery of $\bC_1$  in \eqref{MMV observation} can be formulated as a common SMV CS problem. When $B_{t,1}>1$, it becomes an MMV CS problem \cite{MMV_Tropp}, where the multiple columns of  $\bC_1$ in \eqref{MMV observation} shares a common support.
The optimization problem estimating the row support of $\bC_1$ for MMV is now given by
\begin{align}
    \widehat{\bC}_1  =  \underset{\bC_1}{\mathop{\argmin }}\,\left\| {{\mathbf{Y}}_{1}}- \bPhi_1 \bC_1 \right\|_{F}^{2}\text{~~subject to }{{\left\|\bC_1\right\|}_{r,0}} \le  L,  \label{AoA OMP}
\end{align}
where $\left\|\bC_1\right\|_{r,0}$ is defined as the number of non-zero rows of $\bC_1$.
Using a similar method as the OMP,  the problem in \eqref{AoA OMP} can be solved by simultaneous OMP (SOMP) \cite{sompJ} that is described in Algorithm 1. The output is the estimated support set $\widehat{\Omega}_r$\footnote{
Here, we assume the number of paths is known as a priori for convenience of performance analysis in Section \ref{sec analyze}. When the number of paths is unavailable as a priori, a threshold can be introduced to compare with the power of the residual matrix $\bR^{(l)}$ in Step 8 at each iteration \cite{TroppGreedy,zhang2021successful}. When the power of $\bR^{(l)}$ is less then the threshold, Algorithm \ref{alg_SOMP} terminates, which generates the estimate of number of paths.
}. For notational simplicity, we omit the subscripts in $\bY_1$ and $\bPhi_1$ in Algorithm \ref{alg_SOMP}.

\begin{algorithm} [t]
\caption{Simultaneous OMP: SOMP($\bY,\bPhi,L$)}
\label{alg_SOMP}
\begin{algorithmic} [1]
\STATE Input: Observations $\bY $, measurement matrix  $\bPhi$, sparsity level $L$.
\STATE Initialization: Support set $\widehat{\Omega}^{(0)}= \emptyset$, residual matrix $\bR^{(0)} = \bY$.
\FOR{$l = 1$ to $L$}
\STATE Calculate the coefficient matrix: $\bS = \bPhi^H \bR^{(l-1)}$.
\STATE Select the largest index $\eta = \argmax \limits _{i=1,\cdots,G_r} \lA [\bS]_{i,:} \rA_2$.
\STATE Update the support set: $\widehat{\Omega}^{( l)}= \widehat{\Omega}^{( l-1)} \bigcup \eta$.
\STATE Update the recovery of matrix: $\widehat{\bC}= ([{\bPhi}]_{:,\widehat{\Omega}^{( l)} })^{\dagger}\bY$.
\STATE Update the residual matrix: $\bR^{(l)} =\bY- [{\bPhi}]_{:,\widehat{\Omega}^{( l)} }\widehat{\bC}  $.
\ENDFOR
\STATE Output: $\widehat{\Omega}^{( L)}, \widehat{\bC}$.
\end{algorithmic}
\end{algorithm}

It should be emphasized that the choice of the measurement matrix $\bPhi_1$ and  $\bC_1$  has a profound impact on the recovery performance of SOMP \cite{sompJ}. Observing \eqref{MMV observation}, the TSB $\bF_{b,1}$ is incorporated in $\bC_1$, and the RSB $\bW_{b,1}$ is included in the measurement matrix $\bPhi_1$. Thus, in what follows, the design of RSB $\bW_{b,1}$ and TSB $\bF_{b,1}$, is of interest.
\subsubsection{RSB and TSB Design} \label{sectionStageRSB}
Firstly, we focus on the design of TSB $\bF_{b,1}$. Considering $\bC_1 = \bar{\bH}_a\bar{\bA}_t^H\bF_{b,1}$, in order to guarantee that $\bF_{b,1}$ is unbiased for each item (column) in $\bar{\bA}_t$,
 we design $\bF_{b,1}$ by maximizing the minimum correlation between $\bF_{b,1}$ and each column in $\bar{\bA}_t$, which yields
\begin{align}
   \max_{\bF_{b,1}} \min_{i} {{\| \bF_{b,1}^H{{\left[ \bar{\bA}_t \right]}_{:,i}} \|}_{2}} ~~\text{subject to}~ \bF_{b,1}^H \bF_{b,1} =p_1  \bI_{B_{t,1}}, \label{AoA trans beam}
\end{align}
where $p_1$ is the power allocation of Stage I.
After taking the constraint into account, the optimal solution to the problem in \eqref{AoA trans beam} should ideally satisfy the following
${{\| \bF_{b,1}^H{{\left[ \bar{\bA}_t \right]}_{:,i}} \|}_2}=\sqrt{{p_1B_{t,1}}/{N_t}},~\text{ }i=1,\ldots, G_t$.
It means that $\bF_{b,1}$ is isometric to all columns of $\bar{\bA}_t$, which is obtained by
\begin{align}
\bF_{b,1}=\sqrt{p_1}\left[ {{\mathbf{e}}_{1}},{{\mathbf{e}}_{2}},\ldots ,{{\mathbf{e}}_{B_{t,1}}} \right],\label{F beams AoAs}
\end{align}
where  $\be_i$ the $i$th column of $\bI_{N_t}$.
The construction of $\be_j, ~ j=1, \ldots, B_{t,1}$ in \eqref{F beams AoAs} using the hybrid analog-digital array is possible due to the fact that any vector can be constructed by
linearly combining $N (\geq 2)$ RF chains \cite{xzhang}. To be more specific, there exists $\bF_{A,j}\in \C^{N_t \times N}$, $\boldf_{D,j}\in \C^{N \times 1}$, and $s_j=1$ such that $\be_{j}=\bF_{A,j} \boldf_{D,j}s_j$, i.e.,
\beq
 \be_{j} =  \underbrace{\frac{1}{\sqrt{N_t}} [ \mathbf{1}_{N_t} ~ \tilde{\mathbf{1}}_{N_t}^{(j)} ~  \mathbf{1}_{N_t} \cdots   \mathbf{1}_{N_t} ]  }_{\triangleq \bF_{A,j} }
           \underbrace{\frac{\sqrt{ N_t}}{2} \left[1,-1,0 ,\cdots , 0 \right]^T }_{\triangleq \boldf_{D,j}}\times  1, \label{design e1}
\eeq
where $\tilde{\mathbf{1}}_{N_t}^{(j)} \in \R^{N_t \times 1}$ is defined as the all one vector $\bone_{N_t}\in \R^{N_t\times 1}$ other than the $j$th entry being $-1$.

For the measurement matrix $\bPhi_1 = \bW_{b,1} \bar{\bA}_r$, we optimize $\bW_{b,1}$
by incorporating the isometric CS measurement matrix design criterion \cite{sensingOPT1, sensingOPT2,hadi2015}:
\begin{align}
\min _{\bPhi_1} \lA \bPhi_1 ^H \bPhi_1 - \bI_{G_r}\rA_F^2. \label{sensing opt problem}
\end{align}
After performing standard algebraic manipulations and exploiting the fact $\bar{\bA}_r \bar{\bA}_r^H  = \frac{G_r}{N_r}\bI_{N_r}$, the optimality condition for \eqref{sensing opt problem} is that  the columns of $\bW_{b,1}$ are orthogonal.
Accounting for the analog-digital array constraint into $\bW_{b,1}$ and setting $B_{r,1} = N_r$, we use the DFT matrix $\bS_{N_r} \in \C^{N_r \times N_r}$ such that
\begin{align}
\bW_{b,1} = \bS_{N_r}, \label{RSB 1}
\end{align}
where $[\bS_{N_r}]_{m,n}=\frac{1}{\sqrt{N_r}}e^{-j\frac{2\pi (m-1)(n-1)}{N_r}}, \forall m,n$.

Based on the RSB in \eqref{RSB 1}, in the following, the distribution of the noise term in \eqref{MMV observation} is discussed.
\begin{Proposition}\label{noise semi}
For any semi-orthogonal matrix $\bA\in \C^{m\times n}$ with $\bA \bA^H = \bI$ and random vector $\bn\in \C^{n \times 1}$ with \gls{iid} entries according to $\cC\cN(0, \sigma^2)$, then if we denote $\bb = \bA\bn$, and the entries in $\bb$ are also \gls{iid} $\cC\cN(0, \sigma^2)$.
\end{Proposition}
\begin{proof}
The covariance matrix of $\bb$ is given by $\E[ \bA\bn\bn^H\bA^H] = \sigma^2 \bI$. Because the entries in $\bb$ are obviously complex Gaussian, thus, from the property of Gaussian distribution, the entries in $\bb$ are also \gls{iid} $\cC\cN(0, \sigma^2)$.
\end{proof}

\begin{Remark} \label{remark 2}
Due to the semi-orthogonality of $\bW_{b,1}$ in \eqref{RSB 1}, according to Proposition \ref{noise semi}, the effective noise matrix $\bW_{b,1}^H\bN_1 \in \C^{N_r \times B_{t,1}}$ in \eqref{MMV observation} has i.i.d. Gaussian entries, i.e., $[\bW_{b,1}^H\bN_1]_{i,j} \sim \cC\cN(0,\sigma^2),\forall i,j$. Moreover, since $\bPhi_1 =\bW_{b,1}^H \bar{\bA}_r$, we have $\| [\bPhi_1]_{:,i} \|_2 = 1, \forall i$.
\end{Remark}

The algorithmic procedure estimating AoAs are described in Algorithm \ref{alg_AoAs}.
Given the estimated support set
$\widehat{\Omega}_r$ from Algorithm \ref{alg_SOMP},
the output of Algorithm \ref{alg_AoAs} is the estimated AoA array response matrix $\widehat{\bA}_r = [\bar{\bA}_r]_{:,\widehat{\Omega}_r} \in \C^{N_r \times L}$.
Overall, the number of channel uses for the AoA estimation is
$K_1 = B_{t,1} \frac{ N_r }{N}$.

\begin{algorithm} [t]
\caption{AoA Estimation Algorithm}
\label{alg_AoAs}
\begin{algorithmic} [1]
\STATE Input: Channel dimension $N_r$, $N_t$, number of RF chains $N$, channel paths $L$, power allocation $p_1$, receive array response dictionary $\bar{\bA}_r \in \C^{N_r \times G_r}$.
\STATE Initialization: Generate the TSB $\bF_{b,1} = \sqrt{p_1}[ \be_1,\ldots,\be_{B_{t,1}}]$ in \eqref{F beams AoAs} according to \eqref{design e1} and the RSB $\bW_{b,1} =  \bS_{N_r}$ in \eqref{RSB 1}.
\STATE Collect the observations $\bY_1 = \bW_{b,1}^H \bH \bF_{b,1} + \bW_{b,1}^H\bN_1$.
\STATE
Solve the problem in \eqref{AoA OMP} by using Algorithm \ref{alg_SOMP} with the sparsity level $L$ and $\bPhi_1 = \bW_{b,1}^H \bar{\bA}_r$,
\vspace{-0.1cm}
\begin{align}
(\widehat{\Omega}_r, \widehat{\bC}_1) = \text{SOMP}(\bY_1, \bPhi_1, L). \nonumber
\end{align}
\vspace{-0.5cm}
\STATE Output: Estimation of AoA array response matrix $\widehat{\bA}_r = [\bar{\bA}_r]_{:,\widehat{\Omega}_r}$.
\end{algorithmic}
\end{algorithm}

\subsection{Stage II: AoD Estimation} \label{AoDs method}
To attain the estimation of AoDs, we can utilize the similar method as Stage I. Similar to the one-stage sounding in \eqref{matrix observations}, the observations of Stage II in Fig. \ref{alg diagram} is expressed as
$\bY_2 \in \C^{B_{r,2} \times B_{t,2}}$,
\begin{align}
{{\mathbf{Y}}_{2}}=\mathbf{W}_{b,2}^H \bH {\bF_{b,2}}+\mathbf{W}_{b,2}^H\mathbf{N}_2, \label{observation se}
\end{align}
where  ${\bW_{b,2}}\in {\bC^{N_r\times B_{r,2}}}$ and ${\bF_{b,2}}\in {\bC^{N_t\times B_{t,2}}}$ are the RSB and TSB of the Stage II, respectively.  The $\mathbf{N}_2 \in {\bC^{N_r\times B_{t,2}}}$ is the noise matrix with i.i.d. entries according to $\mathcal{C}\mathcal{N}(0,{{\sigma }^{2}})$. \par
Recall from  \eqref{compact channel} and \eqref{redundant channel estimation}, the channel matrix is rewritten as
\begin{align}
\bH=\bar{\bA}_r{{\bar{\bH}_a}\bar{\bA}_t^H }+\bE. \label{two stage channel}
\end{align}
One can find that $\bar{\bA}_r\bar{\bH}_a\in {\bC^{N_r\times G_t}}$  has $L$ non-zero columns, indexed by $\Omega_t$ with $|\Omega_t|=L$.
Then, plugging \eqref{two stage channel} into  \eqref{observation se} and taking conjugate transpose give
\begin{align}
\mathbf{Y}_{2}^H
  &=\underbrace{\bF_{b,2} ^H \bar{\bA}_{t}}_{\triangleq \bPhi_2}  \underbrace{  \bar{\bH}_a^H  \bar{\bA}_r^H     \bW_{b,2} }_{\triangleq \bC_2}+\bF_{b,2} ^H \bE^H \bW_{b,2} + \bN_2^H\bW_{b,2} \nonumber\\
 &=\bPhi_2{{\mathbf{C}}_{2}}  + \bF_{b,2} ^H \bE^H\bW_{b,2}+\bN_2^H\bW_{b,2}, \label{ob 2nd}
\end{align}
where $\bPhi_2 = \bF_{b,2} ^H \bar{\bA}_{t} \in \C^{B_{t,2}\times G_t}$, and $\bC_2 =  \bar{\bH}_a^H  \bar{\bA}_r^H     \bW_{b,2} \in \C^{G_t \times B_{r,2}}$.
It is straightforward that the ${{\mathbf{C}}_{2}}$ has only $L$ non-zero rows indexed by $\Omega_t$.
Similar to \eqref{AoA OMP} in Stage I, the support set $\Omega_t$ estimation problem can be formulated as
\beq
\widehat{\bC}_2 = \underset{\mathbf{C}_2}{\mathop{\argmin }}\,\left\| \mathbf{Y}_{2}^H-\bPhi_2{{\mathbf{C}}_{2}} \right\|_{F}^{2}\text{ subject to }{{\left\| {{\mathbf{C}}_{2}} \right\|}_{r,0}}\le L,\label{s2 formulation}
\eeq
which is solved by Algorithm \ref{alg_SOMP}. In what follows, the design of RSB ${\bW_{b,2}}$ and TSB ${\bF_{b,2}}$ for Stage II is of interest.

\subsubsection{RSB and TSB Design}
For the design of RSB ${\bW_{b,2}}$,  we leverage the estimated AoAs from Stage I to formulate
\begin{align}
\max_{\bW_{b,2}} \min_{i} {{\| \mathbf{W}_{b,2}^H{{[\widehat{\bA}_r ]}_{:,i}}\|}_{2}}. \label{W2 design}
\end{align}
If $\bW_{b,2}$ is semi-unitary, i.e., $\bW_{b,2}^H \bW_{b,2} = \bI_{B_{r,2}}$, the objective value in \eqref{W2 design} satisfies $\| \mathbf{W}_{b,2}^H{{[ \widehat{\bA}_r ]}_{:,i}} \|_2 \le 1, \forall i$ with the equality holding if
\begin{align}
\cR(\bW_{b,2}) = \cR(\widehat{\bA}_r). \label{W2 design sub}
\end{align}
One can check \eqref{W2 design sub} holds only if $B_{r,2} \ge L$.
Without loss of optimality and to save the number of sounders, we set $B_{r,2} = L$.
One solution to \eqref{W2 design sub} is attained when the columns of $\bW_{b,2}$ are the orthonormal basis of $\widehat{\bA}_r$. For example, we let $\bW_{b,2}$ be the $\bQ$-matrix of the QR decomposition\footnote{The QR decomposition is a decomposition of a matrix $\bA\in \C^{m\times n}$ into the product $\bA = \bQ\bR$ of an orthonormal matrix $\bQ \in \C^{m \times n}$ and an upper triangular matrix $\bR \in \C^{ n\times n}$. } of $\widehat{\bA}_r$ such that
\begin{align}
\bW_{b,2} = \mathop{\mathrm{QR}}(\widehat{\bA}_r), \label{expression Wb2}
\end{align}
where $\mathop{\mathrm{QR}}(\cdot)$ returns the $\bQ$-matrix of a given matrix.

\begin{Remark} \label{remark 3}
Due to the semi-orthogonality of $\bW_{b,2}$ and the conclusions in Proposition \ref{noise semi}, the effective noise matrix $\bW_{b,2}^H\bN_2 \in \C^{B_{r,2} \times B_{t,2}}$ in \eqref{observation se} has \gls{iid} Gaussian entries, i.e., $[\bW_{b,2}^H\bN_2]_{i,j} \sim \cC\cN(0,\sigma^2),\forall i,j$.
\end{Remark}
As for the design of ${\bF_{b,2}}$,  we exploit the isometric CS measurement matrix design criterion,
\begin{align}
\min _{\bPhi_2}\| \bPhi_2 ^H \bPhi_2 - \bI_{G_t}\|_F^2. \label{design of F2}
\end{align}
After similar manipulations as \eqref{sensing opt problem}, the optimality condition for $\bF_{b,2}$ of \eqref{design of F2} is that the columns of $\bF_{b,2}$ are orthogonal.
Then, following the same procedure as \eqref{F beams AoAs} and \eqref{design e1}, we obtain the design of TSP $\bF_{b,2}$ below,
\begin{align}
\bF_{b,2}=\sqrt{p_2}[ {{\mathbf{e}}_1},{{\mathbf{e}}_2},\ldots ,{{\mathbf{e}}_{B_{t,2}} }], \label{expression Fb2}
\end{align}
where $p_2$ is the power coefficient of Stage II.

The algorithmic procedure of estimating AoDs are described in Algorithm \ref{alg_AoDs}.
Provided the estimated support set $\widehat{\Omega}_t$, the output of  Algorithm \ref{alg_AoDs} is the estimated AoD array response matrix  $\widehat{\bA}_t = [\bar{\bA}_t]_{:,\widehat{\Omega}_t} \in \C^{N_t \times L}$. The number of channel uses for the AoD estimation in Stage II is
$K_2 = B_{t,2}$, and the overall number of channel uses for two stages is
\begin{align}
K=K_1 + K_2 =  B_{t,1} \frac{ N_r }{N} + B_{t,2}. \label{number of uses}
\end{align}

\begin{Remark}
Recall that the number of observations for the conventional one-stage channel sounding in Fig. \ref{system diagram} is $\cO(L\cdot \log(G_r G_t/L))$ \cite{DonohoCS}. As a comparison, since the proposed two-stage channel sounding in Fig. \ref{alg diagram} only estimates AoA in Stage I, and estimates AoD in Stage II, the number of required observations
is $\cO(L \cdot   \log (G_r/L))$ in Stage I, and $\cO(L\cdot \log (G_t/L))$ in Stage II. The total number of required observations for the proposed two-stage channel sounding is $\cO(L  \cdot  \log (G_r/L))  +  \cO(L  \cdot  \log (G_t/L)) =  \cO(L  \cdot  \log (G_t G_r/L^2 ) $, which is less than the conventional one-stage sounding.
\end{Remark}

\begin{Remark} \label{RSBHappening}
About happening of the design RSB and TSB,  in Stage I, one can find that the design of RSB in \eqref{RSB 1} and TSB in \eqref{F beams AoAs} are completed before the channel estimation, which are then utilized by the transmitter and receiver.
Like the fact that the training pilots are known for the transmitter and receiver in advance before the task of channel estimation, here we also assume that the TSB and RSB are known as a priori. In Stage II, the TSB ${\bF}_{b,2}$ in \eqref{expression Fb2} is also designed in advance, while the RSB $\bW_{b,2}$ in \eqref{expression Wb2} is designed and employed at the receiver side, which requires no feedback to the transmitter. Overall, the proposed method requires no feedback during the whole procedures of the channel estimation.
\end{Remark}

\begin{algorithm} [t]
\caption{AoD Estimation Algorithm}
\label{alg_AoDs}
\begin{algorithmic} [1]
\STATE Input: Channel dimension $N_r$, $N_t$, number of RF chains $N$, channel paths $L$, power allocation $p_2$, output of AoA estimation $\widehat{\bA}_r$, transmit array response dictionary $\bar{\bA}_t \in \C^{N_t \times G_t}$.
\STATE Initialization: Generate the TSB ${\bF}_{b,2} = \sqrt{p_2}[ \be_{1},\ldots,\be_{B_{t,2}}]$ in \eqref{expression Fb2}
and RSB $\bW_{b,2} = \mathop{\mathrm{QR}}(\widehat{\bA}_r)$ in \eqref{expression Wb2}.
\STATE Collect the observations ${\bY}_2 = \bW_{b,2}^H \bH {\bF}_{b,2} + \bW_{b,2}^H {\bN}_2$.
\STATE Solve the problem in \eqref{s2 formulation} by using Algorithm \ref{alg_SOMP} with the sparsity level $L$ and $\bPhi_2 = \bF_{b,2}^H \bar{\bA}_t$,
\vspace{-0.1cm}
\beq
(\widehat{\Omega}_t, \widehat{\bC}_2) = \text{SOMP}(\bY_2^H,\bPhi_2,L). \nonumber
\eeq
\vspace{-0.5cm}
\STATE Output:  Estimation of AoD array response matrix $\widehat{\bA}_t = [\bar{\bA}_t]_{:,\widehat{\Omega}_t}$.
\end{algorithmic}
\end{algorithm}

\subsection{Channel Estimation} \label{R estimation}

Recalling the channel representation in \eqref{compact channel} and after estimating $\widehat{\bA}_r \in \C^{N_r \times L}$ in Algorithm \ref{alg_AoAs} and $\widehat{\bA}_t \in \C^{N_t \times L}$ in Algorithm \ref{alg_AoDs},  we can express the channel estimate as
\begin{align}
\widehat{\bH} = \widehat{\bA}_r \widehat{\bR} \widehat{\bA}_t^H, \label{expression of estimation}
\end{align}
where $\widehat{\bR} \in \C^{L \times L}$ denotes the estimation of $\diag(\bh)$ in \eqref{compact channel}.
In the following, we will discuss how to obtain the estimate $\widehat{\bR}$.
It is worth noting that unlike \eqref{compact channel} we do not restrict $\widehat{\bR}$ to be a diagonal matrix because of the possible permutations in the columns of $\widehat{\bA}_r$ and $\widehat{\bA}_t$.

Recall the observations of each stage, i.e., $\bY_1= \bW_{b,1}^H \bar{\bA}_r{\bar{\bH}_a}\bar{\bA}_t^H  \bF_{b,1}  +  \bW_{b,1}^H \bE \bF_{b,1} +\bW_{b,1}^H\bN_1$, and $\bY_2= \bW_{b,2}^H \bar{\bA}_r{\bar{\bH}_a}\bar{\bA}_t^H  \bF_{b,2}  +  \bW_{b,2}^H \bE \bF_{b,2} +\bW_{b,2}^H\bN_1$.
Since $\bW_{b,1}^H\bN_1$ and $\bW_{b,2}^H\bN_2$ are \gls{iid} Gaussian, incorporating the expressions of channel estimate in \eqref{expression of estimation}, the estimation of $\widehat{\bR}$ is given by
\begin{align} \nonumber
\widehat{\bR} = \argmin_{\bR}\lA \begin{bmatrix}
  \vect( \bY_1 )\\
  \vect( \bY_2)
\end{bmatrix}  -
\begin{bmatrix}
\vect(  \bW_{b,1}^H \widehat{\bA}_r {\bR} \widehat{\bA}_t^H \bF_{b,1} )\\
 \vect(  \bW_{b,2}^H \widehat{\bA}_r {\bR} \widehat{\bA}_t^H \bF_{b,2})
\end{bmatrix}
\rA_F^2,
\end{align}
where the optimal solution is given by
\begin{align} \nonumber
\vect(\widehat{\bR}) = \left( \bA_1^H \bA_1+\bA_2^H \bA_2\right)^{-1}\left(\bA_1^H  \vect(\bY_1) +\bA_2^H\vect(\bY_2)\right),
\end{align}
where $\bA_1 = (\widehat{\bA}_t^H \bF_{b,1})^T\otimes\bW_{b,1}^H \widehat{\bA}_r \in \C^{N_r B_{t,1}\times L^2}$ and $\bA_2 = (\widehat{\bA}_t^H \bF_{b,2})^T\otimes\bW_{b,2}^H \widehat{\bA}_r\in \C^{L B_{t,2}\times L^2}$. Because $N_r B_{t,1} \gg L^2$ and $B_{t, 2}   \gg   L$, the matrix $\bA_1^H   \bA_1  +  \bA_2^H\bA_2   \in   \C^{L^2    \times L^2}$ is always invertible.
\begin{Remark}
After $\widehat{\bR}$ is estimated, the pairing of AoAs and AoDs can be obtained by selecting positions of the largest $L$ entries in  $\widehat{\bR}$. Then, the path gain $h_l, l=1,2,\cdots,L,$ can be calculated by solving a problem like the oracle estimator in \eqref{oracal estimator}, where the two-stage RSBs and TSBs are utilized.
\end{Remark}
\section{Performance Analysis and Resource Allocation} \label{sec analyze}

In this section, we discuss the reconstruction probability of AoAs and AoDs of the proposed two-stage method in Section \ref{section algorithm}. Moreover, we further enhance the reconstruction performance
by performing power and channel use allocation to each stage.

\subsection{Successful Recovery Probability Analysis}
\subsubsection{SRP of AoA Estimation}
As a starting point, we focus on the SRP of Algorithm \ref{alg_SOMP}.
An SRP bound of SOMP was previously studied in \cite{NoiseSOMP}, where the analysis was based on the restricted isometry property constant of the measurement matrix $\bPhi$.
In this work, we instead analyze the recovery performance of Algorithm \ref{alg_SOMP}, based on the mutual incoherence property (MIP) constant\footnote{The MIP constant of matrix $\bPhi$ is quantified by a variable $\mu = \max_{i\neq j}|\langle [\bPhi]_{:,i},[\bPhi]_{:,j}\rangle|$, where $\langle \cdot, \cdot \rangle$ denotes the inner product.} \cite{DonohoMIP} of $\bPhi$.
 \begin{Lemma}\label{lemma SOMP}
 Suppose $\bC \in \C^{N \times d}$ is a row sparse matrix, where $L$ ($\ll N$) rows of $\bC$, indexed by $\Omega$, are non-zero.
 We consider the observation
 $\bY = \bPhi \bC + \bN$, where $\bY \in \C^{M \times d}$, $\bPhi \in \C^{M \times N}$ is the measurement matrix with $L \leq M \ll N$ and $\| [\bPhi]_{:,i}\|_2=1, \forall i$, and $\bN \in \C^{M \times d}$ is the noise matrix
     with each entry \gls{iid} according to complex Gaussian distribution $\cC\cN(0,\sigma^2)$.
Given that the MIP constant $\mu$ of the measurement matrix $\bPhi$ is $\mu< 1/ (2L-1)$, the SRP of Algorithm \ref{alg_SOMP} satisfies
 \begin{align}
   \text{Pr}(\cV_{S})\ge F_2\left(\frac{{(1-(2L-1)\mu)^2 C_{\text{min}}^2}-4\sigma^2 \mu_{M,d}}{4\sigma^2 \sigma_{M,d}}\right), \label{SOMP prob}
  \end{align}
  where $\cV_{S}$  is the event of successful reconstruction of Algorithm \ref{alg_SOMP}, $C_{\min} = \min\limits _{i\in \Omega}  \lA [\bC]_{i,:} \rA_2$, $\mu_{M,d}  =(M^{1/2}  + d^{1/2}) ^2$, $
\sigma_{M,d} = (M^{1/2}  + d^{1/2}) (M^{-1/2}  + d^{-1/2}) ^{1/3} $, and the function $F_2(\cdot)$\footnote{
  The CDF of Tracy-Widom law \cite{TW1,TW2} $F_2(\cdot)$ is expressed as
  \begin{align}
  F_2(s) = \exp\left( \int_s^{\infty}(x-s)q(x) dx\right), \nonumber
  \end{align}
  where $q(x)$ is the solution of Painlev\'{e} equation of type II:
  \begin{align}
  q''(x)=xq(x)+2q(x)^3,~ q(x) \sim \text{Ai}(x), x \rightarrow \infty, \nonumber
  \end{align}
  where $\text{Ai}(x)$ is the Airy function \cite{TW2, TW1}. To save computational complexity, we admit the table lookup method \cite{dataTW} to obtain the value of $F_2(\cdot)$.
  } is the cumulative distribution function (CDF) of Tracy-Widom law \cite{TW2, TW1}.
 \end{Lemma}

\begin{proof}
 See  Appendix \ref{appendix5-1}.
\end{proof}

\begin{Proposition} \label{with noise p}
Suppose the signal model provided in Lemma \ref{lemma SOMP} and, given the quantization error, the observation model $\bY = \bPhi \bC + \tilde{\bN}$,
where effective noise $\tilde{\bN} = \bE+\bN$ with quantization error $\bE$ and Gaussian noise $\bN$ of \gls{iid} $\cC\cN(0,\sigma^2)$ entries.
If $\mu$ is the MIP constant of the measurement matrix $\bPhi$  with $\mu < 1/(2L-1)$, the SRP of Algorithm \ref{alg_SOMP} is given by
\beq
\text{Pr}(\cV_S) \ge F_2\left(\frac{{\left((1-(2L-1)\mu) C_{\text{min}}-2\| \bE\|_2\right)^2}-4\sigma^2 \mu_{M,d}}{4\sigma^2 \sigma_{M,d}}\right),  \label{prob noise case}
\eeq
where $C_{\min} = \min\limits _{i\in \Omega}  \lA [\bC]_{i,:} \rA_2$,
$\mu_{M,d}  =(M^{1/2}  + d^{1/2}) ^2$, and $\sigma_{M,d} = (M^{1/2}  + d^{1/2}) (M^{-1/2}  + d^{-1/2}) ^{1/3} $.

\begin{proof}
  See Appendix \ref{proof prob noise case}.
\end{proof}
\end{Proposition}

As a direct consequence of Proposition \ref{with noise p}, Theorem \ref{AoAs estimation} below quantifies the SRP of AoA estimation in Algorithm \ref{alg_AoAs}.

 \begin{Theorem} \label{AoAs estimation}
Assume the MIP constant of the measurement matrix $\bPhi_1$ in Algorithm \ref{alg_AoAs} satisfies $\mu_1 < 1/(2L-1)$. Then,
the SRP of Algorithm \ref{alg_AoAs} is lower bounded by
\begin{align}
\text{Pr}(\cA_{S})
&\ge F_2 \left(\frac{   ( 1-(2L-1)\mu_1) \left({h_{\text{min}} \sqrt{\frac{p_1 B_{t,1}}{N_t} }}-2\|\bE_1 \|_2\right)^2-4\sigma^2   \mu_{{N_r,B_{t,1}}}}{4\sigma^2 \sigma_{{N_r,B_{t,1}}}} \right) \nonumber \\
&\approx
  F_2  \left(\frac{   (1-(2L-1)\mu_1) {h_{\text{min}}^2 \frac{p_1 B_{t,1}}{N_t} }-4\sigma^2  \mu_{{N_r,B_{t,1}}}}{4\sigma^2 \sigma_{{N_r,B_{t,1}}}}\right) \label{aoa prob temp} \\
&\triangleq  P_{\text{I}}(p_1, B_{t,1}), \label{aoa prob}
\end{align}
where $\cA_{S}$ is the event of successful reconstruction of AoA,
$h_{\min} = \min_{l\le L} |h_l|$ with $h_l$ being the $l$th entry of $\bh$ in \eqref{compact channel},
$\mu_{{N_r,B_{t,1}}}  =(N_r^{1/2}  + B_{t,1}^{1/2}) ^2$, $
\sigma_{{N_r,B_{t,1}}} = (N_r^{1/2}  + B_{t,1}^{1/2}) (N_r^{-1/2}  + B_{t,1}^{-1/2}) ^{1/3} $, and
$\bE_1 = \mathbf{W}_{b,1}^H\bE{\bF_{b,1}}$.
The approximation in \eqref{aoa prob temp} is obtained by neglecting the quantization term $\bE_1$.
 In \eqref{aoa prob}, the SRP lower bound is substituted as a function of $(p_1, B_{t,1})$.
 \end{Theorem}

\begin{proof}
Recalling the observation model in \eqref{MMV observation} with the TSB and RSB in \eqref{F beams AoAs} and \eqref{RSB 1}, respectively, the effective TSB matrix $\bC_1$ in \eqref{MMV observation} satisfies ${{\| {{[{{\mathbf{C}}_{1}}]}_{{{r}_{l}},:}} \|}_{2}}=\sqrt{\frac{p_1B_{t,1}}{N_t}}\left| {{h}_{l}} \right|$, where ${{r}_{l}} \in \Omega_r$ is the index of the $l$th path of $\bA_r$ in $\bar{\bA}_r$ such that $[\bar{\bA}_r]_{:,r_l}=[\bA_r]_{:,l}$, $l=1, \ldots, L$.
Substituting $C_{\min} = \min \limits _{r_l \in \Omega_r}\lA  [\bC_1]_{r_l,:} \rA_2 =\sqrt{\frac{p_1B_{t,1}}{N_t}}\left| {{h}_{\min}} \right|$ in \eqref{SOMP prob} results in \eqref{aoa prob}, which completes the proof.
\end{proof}
\begin{Remark} \label{aoa fixed p}
According to Theorem \ref{AoAs estimation}, when the power $p_1$ of Stage I is fixed and the number of transmit sounding beams $B_{t,1}$($\ll N_r$) increases, the SRP of AoA increases accordingly. Interestingly, it is more efficient to increase the power allocation $p_1$ than the number of transmit sounding beams $B_{t,1}$ to achieve a higher SRP of AoA.
This can be understood through the two cases where $p_1$ or $B_{t,1}$ grow at the same rate.
Compared to the case of $p_1$, both $\mu_{N_r, B_{t,1}}$ and $\sigma_{N_r, B_{t,1}}$ increase
slowly as $B_{t,1}$ grows, resulting in lower SRP in  \eqref{aoa prob temp}. This aspect will be clearer in the next subsection when we optimize the allocation of $p_1$ and $B_{t,1}$.

\end{Remark}
\subsubsection{SRP of AoD Estimation}
Regarding the SRP of Algorithm \ref{alg_AoDs}, we assume for tractability that the AoA estimation in Stage I was perfect.
The following theorem quantifies the SRP of AoD estimation in Algorithm \ref{alg_AoDs}.

\begin{Theorem} \label{AoDs Prob}
Provided the perfect AoA knowledge known a priori and MIP constant $\mu_2$ of matrix $\sqrt{{N_t}/{(p_2 B_{t,2})}}\bPhi_2$ satisfying $\mu_2 < 1/(2L-1)$, the SRP of Algorithm \ref{alg_AoDs} is lower bounded by
\begin{align}
\text{Pr}(\cD_S)
 &\ge\ F_2 \left(\frac{    (1-(2L-1)\mu_2) h_{\text{min}}  - 2\| \bE_2\|_2)^2 -4\sigma^2 \frac{N_t}{p_2 B_{t,2}}  N_t \mu_{{B_{t,2},L}}}{4N_t\sigma^2 \frac{N_t}{p_2 B_{t,2}} \sigma_{{B_{t,2},L}}}\right) \nonumber\\
  &\approx \ F_2 \left(\frac{    (1-(2L-1)\mu_2)^2 h_{\text{min}}^2   -4\sigma^2 \frac{N_t}{p_2 B_{t,2}}  N_t \mu_{{B_{t,2},L}}}{4N_t\sigma^2 \frac{N_t}{p_2 B_{t,2}}  \sigma_{{B_{t,2},L}}}\right)\label{AoDs pro}\\
 &\triangleq   P_{\text{II}}(p_2, B_{t,2}), \label{AoDs pro temp}
\end{align}
where $\cD_S$ denotes the event of successful AoD reconstruction,
$h_{\min} = \min _{ l\le L} |h_l|$ with $h_l$ being the $l$th entry of $\bh$ in \eqref{compact channel}, $\mu_{{B_{t,2},L}}  =(L^{1/2} + B_{t,2}^{1/2}) ^2$,  $\sigma_{{B_{t,2},L}} = (L^{1/2}  + B_{t,2}^{1/2}) (L^{-1/2}  + B_{t,2}^{-1/2}) ^{1/3}$, and $\bE_2 = \frac{N_t}{p_2 B_{t,2}}\bF_{b,2} ^H \bE\bW_{b,2}$. In \eqref{AoDs pro temp} , the SRP lower bound is substituted as a function of $(p_2, B_{t,2})$.
 \end{Theorem}

\begin{proof}
See Appendix \ref{appendix5-3}.
\end{proof}

\subsection{Power and Channel Use Allocation} \label{section power allocation}
We recall that in the proposed two-stage method, the transmit sounding beams at Stage I and II are, respectively,
$\bF_{b,1} = \sqrt{p_1}[\be_1,\ldots,\be_{B_{t,1}}]$ in \eqref{F beams AoAs} and $\bF_{b,2} = \sqrt{p_2}[\be_{1},\ldots,\be_{B_{t,2}}]$ in \eqref{expression Fb2}.
The total power budget $E$ is therefore defined by
\begin{align}
E = \underbrace{ p_1 B_{t,1}{N_r}/{N}}_{\triangleq E_1} +  \underbrace{ p_2 B_{t,2} }_{\triangleq  E_2}, \label{power value}
\end{align}
where $E_1$ and $E_2$ are the power budgets at the Stage I and Stage II, respectively.

We let  $\eta_1>0$ and $\eta_2>0$ be the target SRP values at Stage I and Stage II, respectively. The SRP-guaranteed power budget minimization problem\footnote{
In \eqref{power allocation problem}, we present the SRP-constrained power minimization problem for optimizing power and channel use allocations. For instance, this criterion can be thought of as a prudent alternative of  the performance maximization subject to power constraints in the MIMO literature because it provides a guarantee on the achievable performance \cite{bengtsson2002pragmatic}.
Multiple variants of the performance-guaranteed power minimization problem can be found in the context of MIMO resource allocation \cite{wiesel2005linear,dahrouj2010coordinated}.
} is  then formulated as
\begin{subequations} \label{power allocation problem}
\begin{align}
&\min_{p_1,p_2,B_{t,1},B_{t,2}} E_1 + E_2    \\
&\text{subject to}~~P_{\text{I}}(p_1,B_{t,1}) \ge \eta_1, ~ P_{\text{II}}(p_2,B_{t,2}) \ge \eta_2,   \\
& ~~~~~~~~~~~E_1 = p_1B_{t,1}{N_r}/{N},~ E_2 =  p_2 B_{t,2},
\\
& ~~~~~~~~~~~B_{t,1}\ge \widetilde{B}_{t,1}, B_{t,2}\ge \widetilde{B}_{t,2},
\end{align}
\end{subequations}
where $\widetilde{B}_{t,1}$ and $\widetilde{B}_{t,2}$ are the minimum numbers of allowed transmit beams at Stage I and Stage II, respectively.
The problem in \eqref{power allocation problem} optimizes the power allocation $p_1$ and $p_2$, and the number of transmit beams $B_{t,1}$ and $B_{t,2}$ to minimizes the total power budget subject to the SRP requirements at Stage I and Stage II. It is worth noting that
that because the problem in \eqref{power allocation problem} is separable, thus \eqref{power allocation problem} is equivalent to the following two  sub-problems,
\begin{subequations} \label{power allocation problem e1}
\begin{align}
&\min_{p_1,B_{t,1}}E_1\\
&\text{subject to}~P_{\text{I}}(p_1,B_{t,1}) \ge \eta_1,E_1 = p_1B_{t,1}\frac{N_r}{N},B_{t,1}\ge \widetilde{B}_{t,1}, \label{power allocation problem e1 b}
\end{align}
\end{subequations}
and
\begin{subequations}\label{power allocation problem e2}
 \begin{align}
&\min_{p_2,B_{t,2}}  E_2\\
& \text{subject to}~ P_{\text{II}}(p_2,B_{t,2}) \ge \eta_2, E_2 =  p_2 B_{t,2},B_{t,2}\ge \widetilde{B}_{t,2}.
\end{align}
\end{subequations}

First of all, we focus on the sub-problem of Stage I in \eqref{power allocation problem e1}.
It is worth noting that directly solving \eqref{power allocation problem e1} is difficult due to the coupled constraints. Thus, we first maximize the SRP, i.e., $P_{\text{I}}(p_1, B_{t,1})$, with arbitrary power budget $E_1$,
\begin{subequations} \label{opt11}
\begin{align}
  &  \max_{p_1,B_{t,1}} P_{\text{I}}(p_1, B_{t,1})
\\
& \text{subject to} ~~ p_1 B_{t,1} N_r/ N = E_1, ~~  B_{t,1} \geq \widetilde{B}_{t,1}.
\end{align}
\end{subequations}
Prior to showing how to solve the problem in \eqref{opt11}, we first elaborate the relation between the problem in \eqref{power allocation problem e1} and \eqref{opt11}. It is easy to observe that as $E_1$ increases the achievable SRP of the objective function in \eqref{opt11} also increases.
Thus, the minimum $E_1$ in \eqref{power allocation problem e1} is achieved when the SRP constraint in \eqref{power allocation problem e1 b}, i.e., $P_{\text{I}}(p_1,B_{t,1}) \ge\eta_1$, holds as the equality. Moreover, given any arbitrary power budget $E_1$ in problem \eqref{opt11}, the interrelation between the power allocation $p_1$ and the number of transmit sounding beams $B_{t,1}$ points to a fundamental tradeoff between them, which is demonstrated in the following theorem.

\begin{Theorem} \label{relation of bt1}
Consider the following non-linear programming
\begin{subequations} \label{opt1}
\begin{align}
(\hat{p}_1, \widehat{B}_{t,1} )  &= \argmax_{p_1,B_{t,1}} P_{\text{I}}(p_1, B_{t,1}) \label{theorem 3a}
\\
& \text{subject to} ~~ p_1 B_{t,1} N_r/ N = E_1, ~~  B_{t,1} \geq \widetilde{B}_{t,1}, \label{theorem 3b}
\end{align}
\end{subequations}
where $E_1$ is an arbitrary power budget. The solution to \eqref{opt1} is given by $\widehat{B}_{t,1} =\widetilde{B}_{t,1}$ and $p_1=\frac{E_1N}{ \widetilde{B}_{t,1} N_r}$.
\end{Theorem}
\begin{proof}
Substituting constraint $p_1 =\frac{E_1N}{{B}_{t,1} N_r}$ in \eqref{theorem 3b}  into the objective function in \eqref{theorem 3a}, we first show that $P_{\text{I}}(\frac{E_1N}{{B}_{t,1} N_r}, B_{t,1})$ in \eqref{theorem 3a} is a monotonically decreasing function of the number of transmit sounding beams $B_{t,1}$ for a fixed $E_1$. Specifically, substituting $\mu_{{N_r,B_{t,1}}}  =(N_r^{1/2}  + B_{t,1}^{1/2}) ^2$ and $
\sigma_{{N_r,B_{t,1}}} = (N_r^{1/2}  + B_{t,1}^{1/2}) (N_r^{-1/2}  + B_{t,1}^{-1/2}) ^{1/3} $ of \eqref{aoa prob temp} into $P_{\text{I}}(\frac{E_1N}{{B}_{t,1} N_r}, B_{t,1})$ gives
\beq
    P_I\left(\frac{E_1N}{{B}_{t,1} N_r}, B_{t,1}\right)=F_2\left(\frac{h_{\text{min}}^2     (1-(2L-1)\mu_1)^2 E_1N - 4N_t N_r\sigma^2 (N_r^{\frac{1}{2}}  + B_{t,1}^{\frac{1}{2}}) ^2 }{4N_tN_r\sigma^2 (N_r^{\frac{1}{2}}  + B_{t,1}^{\frac{1}{2}}) (N_r^{-\frac{1}{2}}  + B_{t,1}^{-\frac{1}{2}}) ^{\frac{1}{3}}} \right).
\label{plug mu sig}
\eeq
Taking the first derivative of the argument inside $F_2(\cdot)$ in \eqref{plug mu sig} with respect to $B_{t,1}$ reveals that the argument is a decreasing function of $B_{t,1}$. This implies that the $P_I(\frac{E_1N}{{B}_{t,1} N_r}, B_{t,1})$ in \eqref{plug mu sig} is a monotonically decreasing function of $B_{t,1}$. Hence, \eqref{opt1} is maximized when $B_{t,1}=\widetilde{B}_{t,1}$, which completes the proof.
\end{proof}

\begin{figure}
\centering
\includegraphics[width=.56\textwidth]{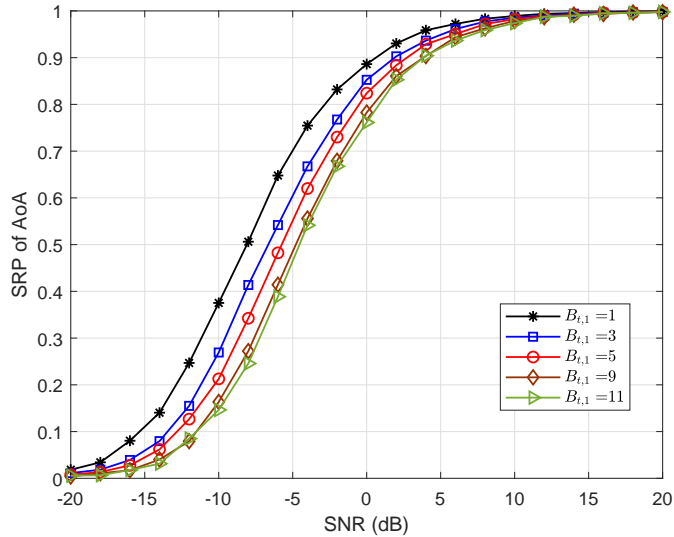}
\caption{SRP of AoA vs. SNR (dB) ($N_r = 20,N_t=64, L=4, N=4, s=1,E_1=10, \widetilde{B}_{t,1}=1$).} \label{AoA_Bt}
\end{figure}

Therefore, based on Theorem \ref{relation of bt1}, the maximum SRP of AoA estimation  for a given $E_1$ is given by
\beq
  P_I\left(\frac{E_1N}{ \widetilde{B}_{t,1} N_r}, \widetilde{B}_{t,1}\right) =
  F_2  \left(\frac{{h_{\text{min}}^2     (1-(2L-1)\mu_1)^2E_1 N /N_r  }-4\sigma^2  N_t \mu_{{N_r, \widetilde{B}_{t,1}}}}{4N_t\sigma^2 \sigma_{{N_r, \widetilde{B}_{t,1}}}}\right) . \label{function va}
\eeq
We demonstrate Theorem \ref{relation of bt1} via numerical simulations in Fig.~\ref{AoA_Bt}, in which the SRP of AoA is evaluated for different numbers of channel uses $B_{t,1}\in\{1,3,5,9,11\}$. The simulation parameters $N_r =20$, $N_t=64$, $L=4$, $N=4$, $s=1$, $E_1=10$, and $\widetilde{B}_{t,1}=1$ are assumed. The curves clearly show that the highest SRP is achieved when $B_{t,1}=1$.

Now, based on Theorem \ref{relation of bt1}, the solution to \eqref{power allocation problem e1} is readily obtained as follows.
In order to make SRP of AoA higher than $\eta_1$ in \eqref{power allocation problem e1},  we solve the inverse function in \eqref{function va} with respect to $E_1$ and conclude that the resource allocation of Stage I should meet the following conditions:
\begin{subnumcases} {\label{resource E1} }
 E_1 = \frac{ 4 \sigma^2 N_t N_r(F_2^{-1}(\eta_1) \sigma_{{N_r,\widetilde{B}_{t,1}}} +  \mu_{{N_r,\widetilde{B}_{t,1}}})}{h_{\text{min}}^2 (1-(2L-1)\mu_1)^2 N},\label{bound E1} \\
   B_{t,1}=\widetilde{B}_{t,1}, \label{bound Bt1}\\
 p_1=  \frac{E_1N}{ \widetilde{B}_{t,1} N_r},\label{bound p1}
\end{subnumcases}
where $F_2^{-1}(\cdot)$ is the inverse function of $F_2(\cdot)$. By using similar procedures of the proof of Theorem \ref{relation of bt1}, we observe the following more general result about the number of vectors $d$ in the signal model stated Lemma \ref{lemma SOMP}.
\begin{Corollary} \label{effect of d}
The bound in \eqref{SOMP prob} is a monotonically decreasing function of the number of measurement vectors $d$.
\end{Corollary}
\begin{Remark} \label{remarkD}
Corollary \ref{effect of d} states the effect of $d$ on the recovery performance of SOMP. It can be interpreted in the following way.
The increase of the number of measurement vectors $d$ has an effect of increasing the number of columns of $\bC$  in Lemma \ref{lemma SOMP} while keeping the $C_{\min}$ unchanged.
This leads to the increase of  the noise power due to the increase in the dimension of $\bN$, which in turn reduces SRP.
\end{Remark}

{
When it comes to the number of channel uses $B_{t,2}$ at Stage II, we cannot reach the same conclusion as Theorem \ref{relation of bt1} because the constant $ \mu_2$ in \eqref{AoDs pro} changes with $B_{t,2}$.
Therefore, Given $B_{t,1} = \widetilde{B}_{t,1}$ and the total number of channel uses $K$ for channel sounding, $B_{t,2}$ is determined
by \eqref{number of uses}, i.e., $K=\widetilde{B}_{t,1}{N_r}/{N} +  B_{t,2}$.
Then, the solution to \eqref{power allocation problem e2} is given by
\begin{subnumcases} {\label{resource E2} }
E_2  = \frac{ 4 \sigma^2 N_t (F_2^{-1}(\eta_2) \sigma_{{B_{t,2},L}} + \mu_{{B_{t,2},L}})}{h_{\text{min}}^2 (1-(2L-1)\mu_2)^2}, \label{bound E2}\\
 B_{t,2}=K-\widetilde{B}_{t,1}{N_r}/{N}\label{bound Bt2},\\
 p_2 = \frac{E_2}{K-\widetilde{B}_{t,1}{N_r}/{N}}. \label{bound p2}
\end{subnumcases}
}

In summary, after solving the two-subproblems in \eqref{power allocation problem e1} and \eqref{power allocation problem e2}, we successfully solve the problem in \eqref{power allocation problem}. The specific resource allocations for two stages are shown in \eqref{resource E1} and \eqref{resource E2}, respectively.
In particular, when the total power budget $E \ge E_{1}+E_{2}$, the joint SRP of AoA and AoD is at least $\eta_1 \eta_2$.

\begin{figure}
\centering
\includegraphics[width=.56\textwidth]{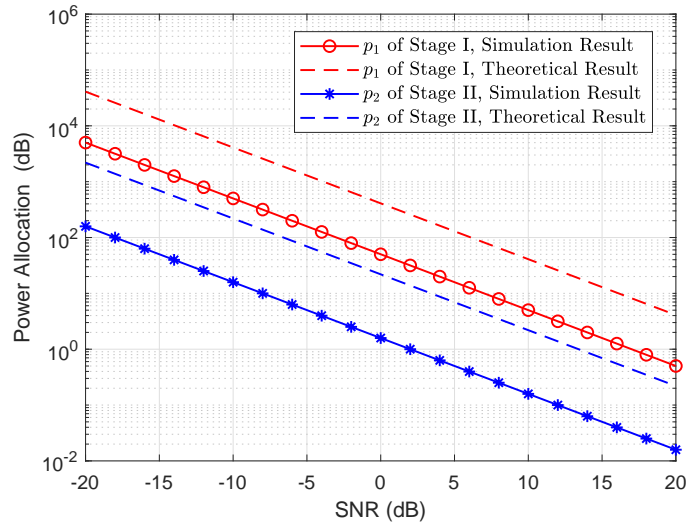}
\caption{{Power allocation to achieve the required SRP vs. SNR (dB) ($N_r = 20,N_t=64, L=4, N=4, s=1, \widetilde{B}_{t,1}=1, \eta_1=\eta_2=0.95$).}} \label{verify resource allocation}
\end{figure}

{
In Fig. \ref{verify resource allocation}, we illustrate the designed resource allocations in \eqref{resource E1} and \eqref{resource E2} with the simulation results.
The parameters are set as $\eta_1 = \eta_2 = 0.95$.
The curves of theoretical results calculate the power allocations $p_1$ and $p_2$ through \eqref{bound p1} and \eqref{bound p2}.
The curves of simulation results are the required power allocations to achieved SRPs of $\eta_1$ and $\eta_2 $.
The simulation parameters $N_r =20$, $N_t=64$, $L=4$, $N=4$, $s=1$ are assumed.
In Fig. \ref{verify resource allocation}, to achieve the same required SRP, i.e., $\eta_1=\eta_2=0.95$, Stage II requires less power allocation than Stage I. This is because the design of the sounding beams for Stage II saves the power consumption. Overall, the trend of the theoretical results is consistent with that of the simulation results, which validates the proposed resource allocation strategies in \eqref{resource E1} and \eqref{resource E2}.
}

\begin{figure}
\centering
\includegraphics[width=.56\textwidth]{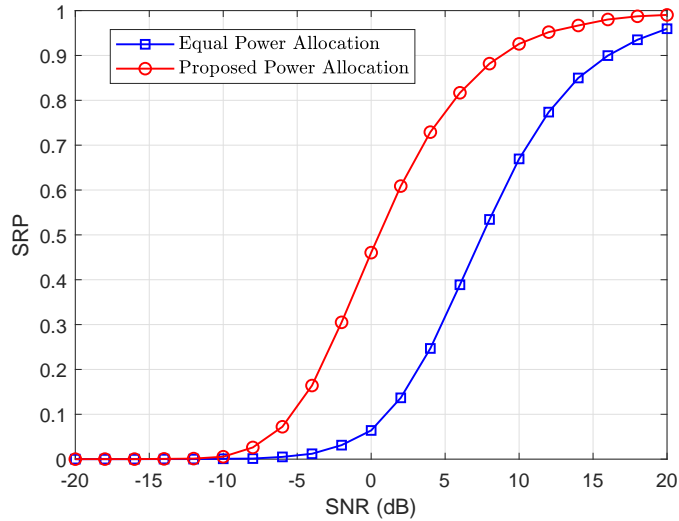}
\caption{Evaluation of the power allocation strategy with  equal power allocation ($N_r = 20,N_t=64, L=4, N=4, s=1, \widetilde{B}_{t,1}=1, \eta_1=\eta_2=0.95$).} \label{verify power allocation}
\end{figure}

In Fig. \ref{verify power allocation}, we demonstrate the SRP of AoA and AoD achieved by the power allocations in \eqref{resource E1} and \eqref{resource E2} compared to the equal power allocation. The power allocations $p_1$ and $p_2$ are calculated by setting $\eta_1=\eta_2=0.95$ and $\sigma=0.1$ in \eqref{resource E1} and \eqref{resource E2}. The simulation parameters are $N_r =20$, $N_t=64$, $L=4$, $N=4$, $s=1$. As we can see from Fig. \ref{verify power allocation}, the proposed power allocation achieves much higher SRP than that of the equal power allocation, which verifies the effectiveness of the proposed power allocation strategy.

\section{Extension to Two-stage Method with Super Resolution} \label{section atomic}

In this section, we extend the proposed two-stage method to the one with super resolution, through which we aim to address the issue of unresolvable quantization errors. Among the existing works, there are two directions to solve the quantization error for off-grid effect. Firstly,  the works in \cite{yang2012off,tang2018off,qi2018off} model the response vector as the summation of on-grid part and the approximation error, in which the sparse Bayesian inference is utilized  to estimate the approximation error. Secondly, the atomic norm minimization has been proposed in \cite{OffGridCS,MmvAtomic,superMM}, which can be viewed as the case when the infinite dictionary matrix is employed. Based on atomic norm minimization, the sparse signal recovery is reformulated as a semidefinite programming. Compared to the sparse Bayesian inference, one advantage of atomic norm minimization is that the recovery guarantee is analyzable \cite{OffGridCS,MmvAtomic,superMM}. Following the methodology of the atomic norm minimization, in this section, we aim to estimate the AoAs and AoDs,  i.e., $\{f_{r,1},\ldots,f_{r,L}\}$ and $\{f_{t,1},\ldots,f_{t,L}\}$, under the proposed two-stage framework.

\subsection{Super Resolution AoA Estimation}
The sounding beams of Stage I, i.e., $\bF_{b,1}$ and $\bW_{b,1}$, are designed according to \eqref{F beams AoAs} and \eqref{RSB 1}. By using the exact expression of $\bH$ in \eqref{compact channel} rather than the quantized version in \eqref{redundant channel estimation}, the observations for Stage I is given by
\begin{align}
\bY_1 &= \bW_{b,1}^H \bH \bF_{b,1} + \bW_{b,1}^H \bN_1 \nonumber \\
 & =\bW_{b,1}^H {{\bA}_{r}}\diag(\mathbf{h})\bA_{t}^{H}{\bF_{b,1}}+\bW_{b,1}^H \bN_1 \nonumber \\
 & =\bW_{b,1}^H {{\bA}_{r}}{{\mathbf{C}}_{r}}+\bW_{b,1}^H \bN_1, \label{music s1}
\end{align}
where $ \bY_1\in\C^{N_r \times B_{t,1}}$ and $\bC_r  =  \diag(\mathbf{h})\bA_{t}^{H}{\bF_{b,1}} \in \C^{L \times B_{t,1}}$.
Since $\bW_{b,1}=\bS_{N_r}$ in \eqref{RSB 1}, projecting $\bY_1$ onto $\bW_{b,1}$ yields
\begin{align}
\widetilde{\bY}_1 = \bW_{b,1}\bY_1 = {{\bA}_{r}}{{\mathbf{C}}_{r}}+ \bN_1. \label{music s1 tilde}
\end{align}

The observation in \eqref{music s1 tilde} is rewritten by explicitly  involving the array response vectors,
\begin{align}
\widetilde{\bY}_1
 =[{\ba_{r}}({{f}_{r,1}}),\ldots ,{\ba_{r}}({{f}_{r,L}})]{{\mathbf{C}}_{r}}+ \bN_1
 =\bR_1 +  \bN_1,
 \label{atomic observation}
\end{align}
where $\bR_1 =[{\ba_{r}}({{f}_{r,1}}),\ldots ,{\ba_{r}}({{f}_{r,L}})]{{\mathbf{C}}_{r}} \in \C^{N_r  \times B_{t,1}}$.
The atom ${{\bA}_{r}}(f,\mathbf{b})\in {\bC^{N_r\times B_{t,1}}}$ is defined in \cite{MmvAtomic,OffGridCS} as
${{\bA}_{r}}(f,\mathbf{b})={\ba_{r}}(f){{\mathbf{b}}^{H}}$, where $f\in [0,1)$ and $\bb \in \C^{B_{t,1} \times 1}$ with ${{\left\| \mathbf{b} \right\|}_{2}}=1$.
We let the collection of all such atoms be the set
$\mathcal{A}_{r}=\{{{\bA}_{r}}(f,\mathbf{b}): f\in [0,1), {{\left\| \mathbf{b} \right\|}_{2}}=1\}$.
Obviously, the cardinality of $\mathcal{A}_r$ is infinite. 
The matrix ${{\mathbf{R}}_{1}}$ in \eqref{atomic observation} can be written as the linear combination among the atoms from the atomic set ${{\mathcal{A}}_{r}}$,
\begin{align}
{{\mathbf{R}}_{1}}=\sum\limits_{l=1}^{L}{{{[\mathbf{c}_r]}_{l}}{{\bA}_{r}}({{f}_{r,l}},{{\mathbf{b}}_{l}})}=\sum\limits_{l=1}^{L}{{{[\mathbf{c}_r]}_{l}}{\ba_{r}}({{f}_{r,l}})}\mathbf{b}_{l}^{H},
\label{atom repre}
\end{align}
where $\bc_r \in \R^{L \times 1}$ is the coefficient vector with $[\bc_r]_l \ge 0$, and it has the relationship $[\bC_r]_{l,:} = [\bc_r]_l \bb_l^H,~ \forall l = 1, \ldots, L$.
Observing \eqref{atom repre}, the dimension of vector $\bc_r$, i.e., $L$, can be interpreted as the sparest representation of $\bR_1$
in the context of the atomic set $\cA_r$. Therefore, in order to seek the sparsest representation,
after taking the noise in \eqref{atomic observation} into account, the reconstruction problem is formulated by
\begin{align}
	\underset{\bR_1}{\mathop{\min }}\,{\left\| \bR_1 \right\|}_{\mathcal{A}_r,0}+ \frac{\lambda_1}{2}  \| \widetilde{\bY}_1-\bR_1\|_F^2, \label{super L0}
\end{align}	
where $\lambda_1>0$ is the penalty parameter, and ${\left\| {{\mathbf{R}}_{1}} \right\|}_{{{\mathcal{A}}_{r}},0}$ is defined as
\begin{subequations}\label{L0 atomic}
\begin{align}
   {{\left\| {{\mathbf{R}}_{1}} \right\|}_{{{\mathcal{A}}_{r}},0}}&=\underset{\mathbf{c}_r}{\mathop{\inf }}\,{{\left\| \mathbf{c}_r \right\|}_{0}} \\
  \text{subject to }&{{\mathbf{R}}_{1}}=\sum\limits_{l=1}^{L}{{{[\mathbf{c}_r]}_{l}}{{\bA}_{r}}({{f}_{r,l}},{{\mathbf{b}}_{l}})}, \\
 & \text{ }{{\bA}_{r}}({{f}_{r,l}},{{\mathbf{b}}_{l}})\in {{\mathcal{A}}_{r}},{{[\mathbf{c}_r]}_{l}}\ge 0,
\end{align}
\end{subequations}
with $\| \bR_1\|_{\cA_r,0}$ revealing the minimal number of atoms in $\bR_1$.
When the sparest representation of $\bR_1$, i.e., $\{{{{[\mathbf{c}_r]}_{l}}{\ba_{r}}({{f}_{r,l}})}\mathbf{b}_{l}^{H}\}_{l=1}^L$, is found by solving \eqref{super L0}, the AoAs $\{f_{r,l}\}_{l=1}^L$
can be obtained from the atomic decomposition in \eqref{atom repre}.
However, since the minimization problem in \eqref{L0 atomic} is combinatorial, it is not tractable to calculate the value of ${{\left\| {{\mathbf{R}}_{1}} \right\|}_{{{\mathcal{A}}_{r}},0}}$. To overcome the challenge, the problem in \eqref{super L0} is relaxed as,
\begin{align}
	\underset{\bR_1}{\mathop{\min }}\,{\left\| \bR_1 \right\|}_{\mathcal{A}_r,1}+ \frac{\lambda_1}{2} \| \widetilde{\bY}_1-\bR_1 \|_F^2, \label{robust L1 s1}
\end{align}	
where $ {{\left\| {{\mathbf{R}}_{1}} \right\|}_{{{\mathcal{A}}_{r}},1}}$ is the atomic norm of $\bR_1$ defined by
\begin{subequations} \label{L1 atomic}
\begin{align}
   {{\left\| {{\mathbf{R}}_{1}} \right\|}_{{{\mathcal{A}}_{r}},1}}&=\underset{\mathbf{c}_r}{\mathop{\inf }}\,{{\left\| \mathbf{c}_r \right\|}_{1}} \\
  \text{subject to }&{{\mathbf{R}}_{1}}=\sum\limits_{l=1}^{L}{{{[\mathbf{c}_r]}_{l}}{{\bA}_{r}}({{f}_{r,l}},{{\mathbf{b}}_{l}})}, \\
 & \text{ }{{\bA}_{r}}({{f}_{r,l}},{{\mathbf{b}}_{l}})\in {{\mathcal{A}}_{r}},{{[\mathbf{c}_r]}_{l}}\ge 0.
\end{align}
\end{subequations}
It is noted that in \eqref{L1 atomic}, the atomic norm $\|\bR_1 \|_{\cA_r,1}$ is to minimize the summation of entries in $\mathbf{c}_r$ instead of the number of non-zero elements in  \eqref{L0 atomic}.

Different from the intractability of \eqref{L0 atomic}, the problem in \eqref{L1 atomic} can be efficiently solved by semi-definite programming \cite{OffGridCS}:
\begin{subequations} \label{atom eq}
\begin{align}
  & {{\left\| {{\mathbf{R}}_{1}} \right\|}_{{{\mathcal{A}}_{r}},\text{1}}}= \underset{\bu,\bZ}{\mathop{\inf }}\,\frac{1}{2}\text{tr}\left( \text{Toeplitz}(\mathbf{u}) \right)+\frac{1}{2}\text{tr}(\mathbf{Z})  \\
 & \text{subject to }
 \left[
 \begin{matrix}
   \text{Toeplitz}(\mathbf{u}) & {{\mathbf{R}}_{1}}  \\
   \mathbf{R}_{1}^H & \mathbf{Z}  \\
\end{matrix}
\right]\succeq \mathbf{0},
\end{align}
\end{subequations}
where $\mathbf{u}\in {\bC^{N_r\times 1}},\mathbf{Z}\in {\bC^{B_{t,1}\times B_{t,1}}}$, and $\text{Toeplitz}(\mathbf{u})\in {\bC^{N_r\times N_r}}$ denotes the Hermitian Toeplitz matrix generated by the vector $\bold{u}$. Plugging \eqref{atom eq} into \eqref{robust L1 s1} gives
\begin{subequations} \label{robust L1 s1 T}
\begin{align}
  & \underset{\bu,\bZ,\bR_1}{\mathop{\inf }} ~ \text{tr}\left( \text{Toeplitz}(\mathbf{u}) \right)+\text{tr}(\mathbf{Z}) + {\lambda_1}  \| \widetilde{\bY}_1-\bR_1 \|_F^2\\
& \text{subject to }
\bX = \left[
  \begin{matrix}
   \text{Toeplitz}(\mathbf{u}) & {{\mathbf{R}}_{1}}  \\
   \mathbf{R}_{1}^H & \mathbf{Z}  \\
\end{matrix}
\right], ~  \bX \succeq \mathbf{0}.
\end{align}
\end{subequations}
It is straightforward to find that \eqref{robust L1 s1 T} is convex, where ADMM can be employed to accelerate the computation.
The augmented Lagrangian of \eqref{robust L1 s1 T} is expressed as
\begin{align}
\mathcal{L}(\bu,\bZ,\bR_1,\bX, \bLambda)
&=\text{tr}\left( \text{Toeplitz}(\mathbf{u}) \right)+\text{tr}(\mathbf{Z}) + {\lambda_1}  \| \widetilde{\bY}_1- \bR_1 \|_F^2 \nonumber  \\
&~~+ \left< \bLambda, \bX- \left[\begin{matrix}
   \text{Toeplitz}(\mathbf{u}) & {{\mathbf{R}}_{1}}  \\
   \mathbf{R}_{1}^H & \mathbf{Z}  \\
\end{matrix}
\right]\right> + \frac{\rho}{2}\lA  \bX- \left[\begin{matrix}
   \text{Toeplitz}(\mathbf{u}) & {{\mathbf{R}}_{1}}  \\
   \mathbf{R}_{1}^H & \mathbf{Z}  \\
\end{matrix}
\right] \rA_F^2, \label{aug lag}
\end{align}
where $\bX \in \C^{(N_r + B_{t,1}) \times (N_r + B_{t,1})}$ and $\bLambda  \in \C^{(N_r + B_{t,1}) \times (N_r + B_{t,1})}$ are Hermitian matrices, and $\rho$ is the Lagrangian multiplier. Then, with $t$ being the iteration index, we iteratively update the variables in \eqref{aug lag} as follows:
\begin{align}
           (\bu^{t+1},\bZ^{t+1},\bR_1^{t+1})   &= \argmin_{\bu,\bZ, \bR_1} \mathcal{L}(\bu, \bZ,\bR_1,\bX^t, \bLambda^{t}) ,\label{admm s1}\\
\bX^{t+1} &= \argmin_{\bX\succeq \mathbf{0}} \mathcal{L}(\bu^{t+1} ,\bZ^{t+1} ,\bR_1^{t+1} ,\bX , \bLambda^{t}),\label{admm s2}\\
\bLambda^{t+1} &= \bLambda^{t} + \rho\left( \bX^{t+1}- \left[\begin{matrix}
   \text{Toeplitz}(\mathbf{u}^{t+1}) &   {{\mathbf{R}}_{1}^{t+1}}  \\
   (\mathbf{R}^{t+1}_{1})^H &   \mathbf{Z}^{t+1}
\end{matrix}
\right]  \right).
\end{align}
The solutions of the \eqref{admm s1} and  \eqref{admm s2} are respectively
\begin{align}
[\bu^{t+1}]_i &=
\begin{cases}
\frac{V_i+\rho S_i}{(N_r -t)\rho+N_r}, &i=1\\
\frac {V_i+\rho S_i}{(N_r -t)\rho}, & i=2,\ldots,N_r
\end{cases},\text{with } V_i=\sum_{k=1}^{N_r+1-i}[\bLambda^t]_{k,k-1+i}, ~S_i=\sum_{k=1}^{N_r+1-i}[\bX^t]_{k,k-1+i},\nonumber\\
  \bR_1^{t+1}& = \frac{1}{ \lambda_1+\rho}  (\lambda_1 \widetilde{\bY}_1+ \rho[\bX^{t}]_{1:N_r, N_r + 1:\text{end}}+[\bLambda^{t}]_{1:N_r ,N_r + 1:\text{end}}),\nonumber\\
\bZ^{t+1} &=   \frac{1}{\rho}([\bLambda^{t}]_{N_r+1:\text{end},N_r+1:\text{end}}+ \rho[\bX^{t}]_{N_r+1:\text{end},N_r+1:\text{end}}-\bI_{B_{t,1}}),\nonumber\\
\bX^{t+1}&=
\left[\begin{matrix}
   \text{Toeplitz}(\mathbf{u}^{t+1}) & {{\mathbf{R}}_{1}^{t+1}}  \\
   (\mathbf{R}^{t+1}_{1})^H & \mathbf{Z}^{t+1}  \\
\end{matrix}
\right] -\frac{1}{\rho} \bLambda^t. \nonumber
\end{align}
It is worth noting that in order to guarantee $\bX \succeq 0$ as shown in \eqref{admm s2}, we can set the negative eigenvalues of $\bX^{t+1}$ to $0$.  When the iterative process converges,
the result $\text{Toeplitz}(\bu)$ can be utilized to obtain the estimation of AoAs.
Specifically, we can take Vandermonde decomposition \cite{OffGridCS} for $\text{Toeplitz}(\bu)$,
$\text{Toeplitz}(\bu) = \bV \bD \bV^H$,
where
$\bV=[\ba_r(\hat{f}_{r,1}),\ldots,\ba_r(\hat{f}_{r,L})] \in \C^{N_r \times L}$ with $\{\hat{f}_{r,l} \}_{l=1}^L$ being the estimated AoAs and $\bD = \diag([d_1,\ldots,d_L])\in \C^{L \times L}$.
In practice, it is not necessary to calculate the Vandermonde decomposition of $\text{Toeplitz}(\bu)$ explicitly. Since the column subspace of $\text{Toeplitz}(\bu)$ is equal to $\cR(\bV)$,
the set of AoAs can be estimated from $\text{Toeplitz}(\bu)$ efficiently by spectrum estimation algorithms such as MUSIC or ESPRIT  \cite{MmvAtomic,OffGridCS}.
\subsection{Super Resolution AoD Estimation}
Similarly, the observations for the second stage is given by
\begin{align}
\bY_2 &= \bW_{b,2}^H \bH \bF_{b,2} + \bW_{b,2}^H \bN_2 \nonumber \\
 & =\bW_{b,2}^H{{\bA}_{r}}\diag(\mathbf{h})\bA_{t}^{H}\bF_{b,2}+\bW_2^H \bN_2  \nonumber\\
 & =\bC_t \bA_t^H \bF_{b,2}+\bW_{b,2}^H \bN_2, \label{music s2}
 \end{align}
 where we let $\bC_t = \bW_{b,2}^H{{\bA}_{r}}\diag(\mathbf{h}) \in \C^{L \times L }$.
At Stage II, the observation $\bY_2$ in \eqref{music s2} is rewritten as
\begin{align}
\bY_2^H
 = \bF_{b,2}^H \bA_t \bC_t^H + \bN_2^H\bW_{b,2} =\bR_2^H +\bN_2^H \bW_{b,2} ,
\label{atomic observation 2}
\end{align}
where we let $\bR_2 =\bF_{b,2}^H \bA_t \bC_t^H \in \C^{B_{t,2} \times L}$. Due to the design of $\bF_{b,2}$ in \eqref{expression Fb2}, we have
\begin{align}
\bF_{b,2}^H \bA_t =\sqrt{p_2} [\bA_t]_{1:B_{t,2},:}
=\sqrt{p_2}[\mathbf{a}_t({{f}_{t,1}}),\ldots ,\mathbf{a}_t({{f}_{t,L}})]_{1:B_{t,2},:} . \nonumber
\end{align}
For convenience, we define $\widetilde{\ba}_t(f) =[\ba_t(f)]_{1:B_{t,2}} \in \C^{B_{t,2} \times 1}$ and $\widetilde{\bA}_t = [\widetilde{\ba}_t({{f}_{t,1}}),\ldots ,\widetilde{\ba}_t({{f}_{t,L}})]\in \C^{B_{t,2} \times L}$.
The AoD estimation boils down to extracting $L$ parameters $\left\{ {{f}_{t, l}} \right\}_{l=1}^{L}$ in $\widetilde{\bA}_t$. We let ${{\bA}_{t}}(f,\mathbf{b})\in {\bC^{B_{t,2}\times L }}$ be
${\bA_{t}}(f,\mathbf{b})={\widetilde{\ba}_{t}}(f){{\mathbf{b}}^{H}}$, where
$f\in [0,1)$, $\bb \in \C^{L \times 1}$ with ${{\left\| \mathbf{b} \right\|}_{2}}=1$, and the atomic set $\cA_t$ is defined by
$\mathcal{A}{_{t}}= \{{{\bA}_{t}}(f,\mathbf{b}): f\in [0,1], {{\left\| \mathbf{b} \right\|}_{2}}=1\}$,
Similarly, ${\mathbf{R}}_{2}^H$ in \eqref{atomic observation 2} can be written as the linear combination of the atoms from the set ${\mathcal{A}}_{t}$,
\begin{align}
{\mathbf{R}}_{2}^H =\sum\limits_{l=1}^{L}{{[\mathbf{c}_t]}_l}{{\bA}_{t}}({{f}_{t,l}},{{\mathbf{b}}_{l}})=\sum\limits_{l=1}^{L}{{{[\mathbf{c}_t]}_{l}}{\ba_{t}}({{f}_{t,l}})}\mathbf{b}_{l}^{H},
\nonumber
\end{align}
where $\bc_t \in \R^{L \times 1}$ is the coefficient vector with $[\bc_t]_l \ge 0$.
Therefore, using the similar approach as AoA estimation in \eqref{robust L1 s1}, the AoD estimation problem is given by
\begin{align}
\underset{\bR_2}{\mathop{\min }}\,{\left\| \bR_2^H  \right\|}_{\mathcal{A}_t,1}+ \frac{\lambda_2}{2}  \left\| \bY_2-\bR_2 \right\|_F^2,  \label{robust L1 s2}
\end{align}
where $\lambda_2$ is a penalty parameter.
The problem in \eqref{robust L1 s2} can also be solved in a similar manner as \eqref{robust L1 s1}, and the estimation for AoDs, i.e., $\{\hat{f}_{t,l}\}_{l=1}^L$, can be obtained.

Furthermore, after the AoAs $\{\hat{f}_{r,l}\}_{l=1}^L$ and AoDs $\{\hat{f}_{t,l}\}_{l=1}^L$ are estimated, we can easily calculate the AoA and AoD array response matrix $\widehat{\bA}_r$ and $\widehat{\bA}_t$. Then, by using the channel estimation technique provided in Section \ref{R estimation}, the final channel estimation result is obtained.

\section{Simulation Results} \label{section simulation}
In this section, we evaluate the performance of the proposed two-stage AoA and AoD estimation method and two-stage method with super resolution.
For comparison, we take the OMP-based mmWave channel estimation method \cite{OMPchannel} as our benchmark.
Also, we include the oracle estimator as we discussed in \eqref{oracal estimator}.
The parameter settings for evaluation are as follows.
We assume throughout the simulation $N_r=20,N_t=64$, and the channel model is given by \eqref{channel model}. We let the dimensions of the angle grids for the proposed \text{two-stage} method and
OMP \cite{OMPchannel} be $G_r=sN_r$ and $G_t=sN_t$. The number of paths is $L=4$. The variance of the path gain is $\sigma_l^2 = 1, \forall l$.
The number of RF chains is $N=4$.
The number of channel uses for the estimation task is $K=50$. The minimum allowed transmit beams at Stage I are $\widetilde{B}_{t,1}=1$.
Without loss of generality, for the proposed two-stage framework, the power budget $E = E_1 + E_2$, where $E_1$ and $E_2$ are, respectively,  given by the resource allocations in \eqref{resource E1} and \eqref{resource E2} with $\eta_1=\eta_2=0.95$ and SNR$=20$dB.

To evaluate the estimation performance, we use three performance metrics:
\begin{itemize}
\item
The first metric is the SRP.
The error of the estimated angles are defined as
\begin{align}
\epsilon =\frac{1}{2L} \sum_{l=1}^{L}\left(| f_{r,l}- \hat{f}_{r,l}|^2+| f_{t,l}- \hat{f}_{t,l}|^2\right). \nonumber
\end{align}
We declare the reconstruction is successful if $\epsilon \le 10^{-3}$. Precisely, SRP is defined as
\begin{align}
\text{SRP} = \frac{\text{number of trials with } \epsilon \le 10^{-3}}{\text{number of total trials}}. \nonumber
\end{align}
\item
The second metric is MSE of angle estimation defined as
\begin{align*}
\text{MSE} = \E\left[\sum_{l=1}^{L}\left(| f_{r,l}- \hat{f}_{r,l}|^2+| f_{t,l}- \hat{f}_{t,l}|^2\right)\right].
\end{align*}
\item
The third metric is NMSE of channel estimation defined as
$$\text{NMSE} = \E[\| \bH - \widehat{\bH}\|_F^2/\| \bH \|_F^2],$$
where $\widehat{\bH}$ is the channel estimate.
\end{itemize}

\begin{figure}
\centering
\includegraphics[width=.56\textwidth]{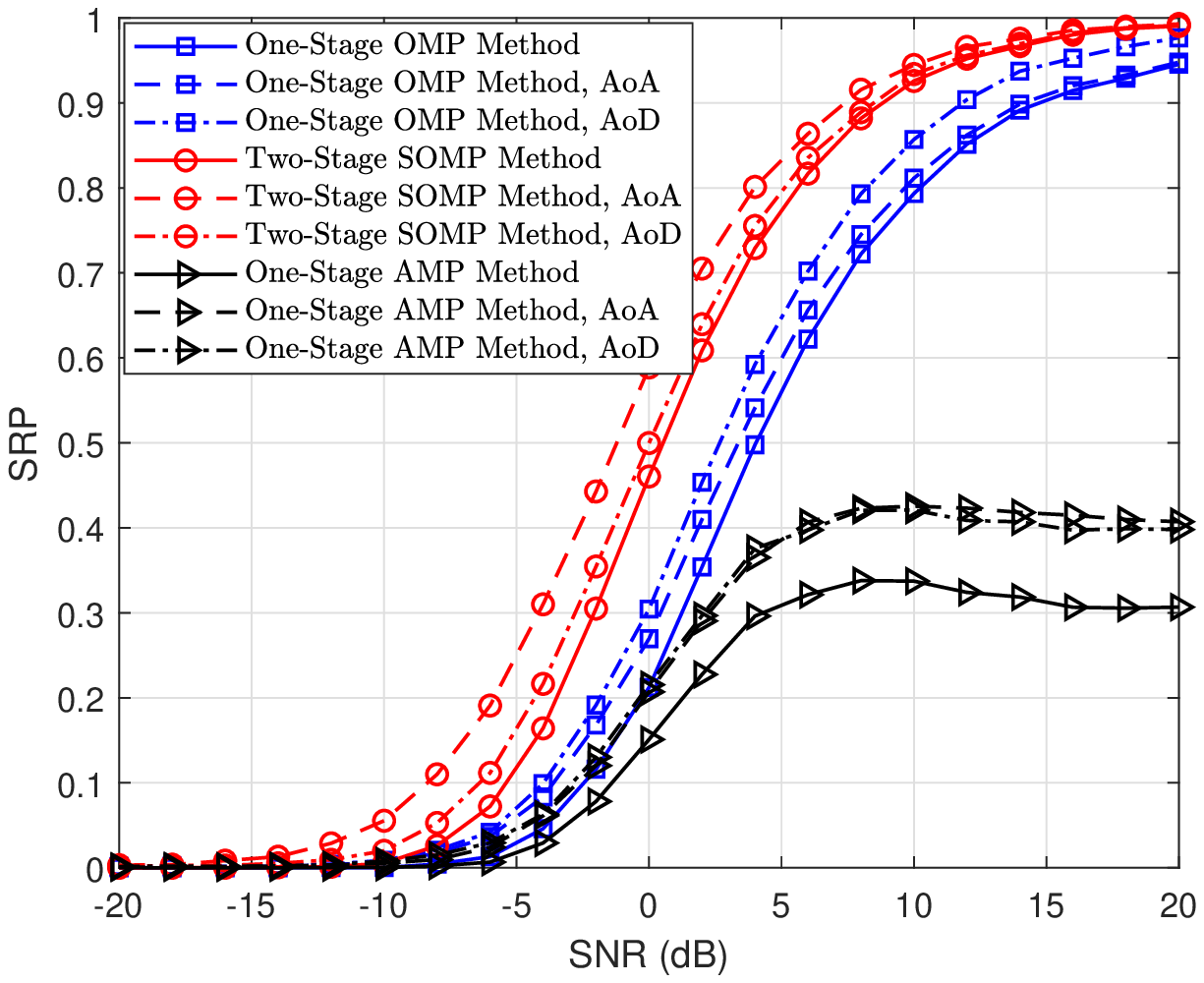}
\caption{SRP vs. SNR (dB) with discrete angles ($N_r = 20,N_t=64, L=4, N=4, K=50,\widetilde{B}_{t,1}=1, s=1$).} \label{bench_pro}
\end{figure}

\begin{figure}
\centering
\includegraphics[width=.56\textwidth]{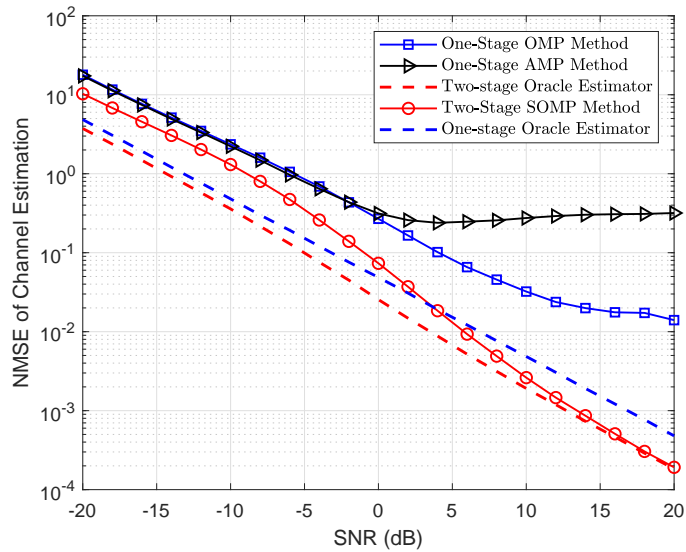}
\caption{NMSE vs. SNR (dB)  with discrete angles  ($N_r = 20,N_t=64, L=4, N=4, K=50,\widetilde{B}_{t,1}=1, s=1$).} \label{bench_NMSE}
\end{figure}

\subsection{Channel Estimation Performance of Two-stage Method with Discrete Angles} \label{section_discrete_simulation}
For the simulations with discrete angles in Figs. \ref{bench_pro}-\ref{bench_NMSE}, the ${{f}_{t,l}}$ and $f_{r,l}$ in \eqref{channel model} are uniformly distributed on the grids of size $G_t=N_t$ and $G_r=N_r$, respectively. Three methods are compared, which are proposed two-stage SOMP method, one-stage OMP method \cite{OMPchannel}, AMP method \cite{donoho2009message}, and oracle method in \eqref{oracal estimator}. We show the SRP in Fig. \ref{bench_pro} and NMSE in Fig. \ref{bench_NMSE}.\par
In Fig. \ref{bench_pro}, considering that oracle method assumes that AoAs and AoDs are known as a priori, we do not illustrate the performance of the oracle method when comparing the SRP. As can be seen in Fig. \ref{bench_pro}, the proposed two-stage SOMP method achieves a higher SRP compared to the benchmarks. It is worth noting that the AMP-based method require the minimal measurements to guarantee the convergence \cite{donoho2009message}. When the number of channel uses is limited, the AMP-based method can not achieve the near one SRP even if the SNR is high. Also, the SRPs of AoA and AoD of the proposed two-stage SOMP method are both higher than those of one-stage OMP method. The improvement of SRP of AoD is because we optimize the sounding beams of the second stage based on the estimated AoA result. For the improvement of SRP of AoA, it is because we allocate more power budget to Stage I according to the proposed resource allocation strategy.\par
Similarly, in Fig. \ref{bench_NMSE}, the proposed two-stage SOMP method has lower NMSE than the one-stage OMP and AMP methods. In addition, we can find from Fig. \ref{bench_NMSE} that the proposed two-stage SOMP method converges to the performance of the oracle method as SNR grows. Overall, Figs. \ref{bench_pro}-\ref{bench_NMSE} verify that the proposed two-stage method outperforms the one-stage OMP in the scenario of discrete angles.

\subsection{Channel Estimation Performance of Two-stage Method with Continuous Angles} \label{sec_cont_simulation}
For this set of simulations in Fig. \ref{bench_CON_Angle_Error}-\ref{bench_CON_H_Error}, we assume the ${{f}_{t,l}}$ and $f_{r,l}$ in \eqref{channel model} are uniformly distributed in $[0,1)$. Four methods are compared, which are the proposed two-stage SOMP method, two-stage method with super resolution, one-stage OMP method \cite{OMPchannel}, and one-stage atomic method \cite{superMM}. When implementing the two-stage SOMP method and one-stage OMP method with the defined angle grids, the estimated angles are located on the defined grids. Fig. \ref{bench_CON_Angle_Error} illustrates the MSE and Fig. \ref{bench_CON_H_Error} illustrates the NMSE of channel estimation.\par
In Fig. \ref{bench_CON_Angle_Error}, the proposed two-stage SOMP method and two-stage method with super resolution outperform the one-stage OMP and one-stage atomic method, respectively. Interestingly, the two-stage SOMP method achieves the minimal MSE when SNR is low. This is because when SNR is low, i.e., $\text{SNR}\le 5\text{dB}$, the noise power is higher than that of the quantization error. Therefore, using the quantized model could reduce the complexity of problem and achieve near-optimal performance. When SNR is high, i.e., $\text{SNR}\ge 5\text{dB}$, the two-stage method with super resolution achieves the minimal MSE. This is because when SNR is high, the quantization error will become dominant, which can not be handled by the grid-based methods. Nevertheless, the Fig. \ref{bench_CON_Angle_Error} verifies that by dividing the estimation into two stages, the estimation of AoAs and AoDs is improved compared to the one-stage estimation.\par
Likewise, in Fig. \ref{bench_CON_H_Error}, the proposed  two-stage SOMP method  and two-stage method with super resolution also achieve lower NMSE than the one-stage OMP and one-stage atomic methods.  Similarly, when SNR is high, the  two-stage method with super resolution shows the minimum NMSE.

\begin{figure}
\centering
\includegraphics[width=.56\textwidth]{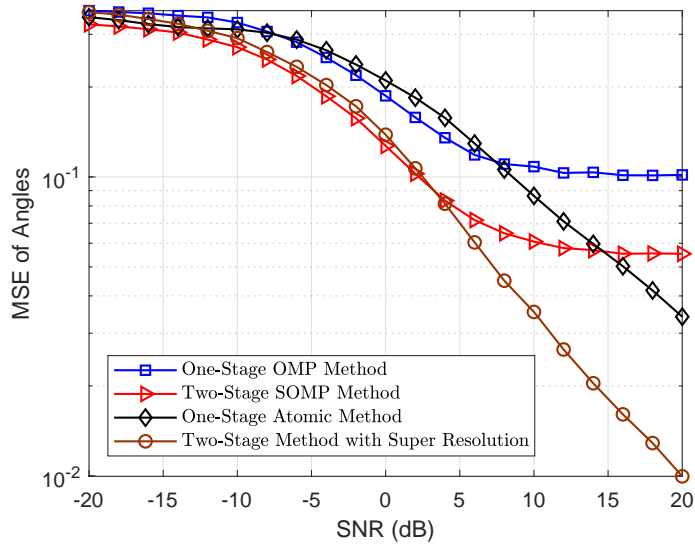}
\caption{MSE vs. SNR (dB) with continuous angles ($N_r = 20,N_t=64, L=4, N=4, K=50,\widetilde{B}_{t,1}=1,s=2$).} \label{bench_CON_Angle_Error}
\end{figure}

\begin{figure}
\centering
\includegraphics[width=.56\textwidth]{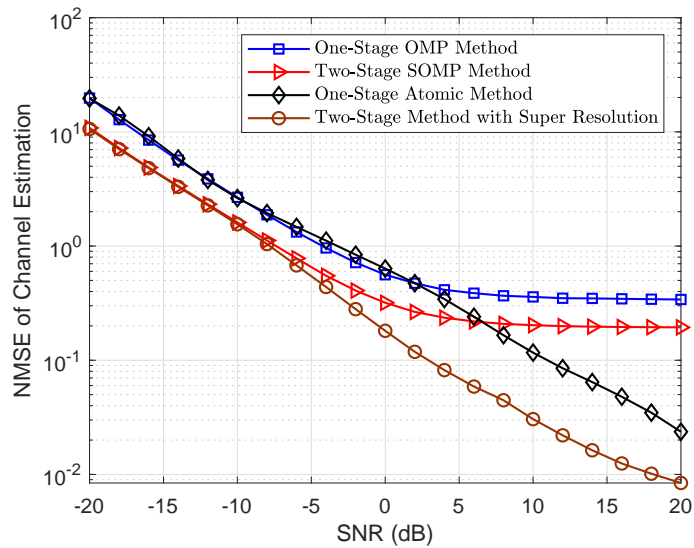}
\caption{NMSE vs. SNR (dB) with continuous angles ($N_r = 20,N_t=64, L=4, N=4, K=50,\widetilde{B}_{t,1}=1,s=2$).} \label{bench_CON_H_Error}
\end{figure}

\subsection{Analysis of Computational Complexity} \label{Section Complexity}
For two-stage method, the computational complexity for the first stage is $\mathcal{O}(L N_r G_r) = \mathcal{O}(sL N_r^2)$, and the complexity for the second stage is $\mathcal{O}(LB_{t,2} G_t)=\mathcal{O}(sL(K-N_r/N)N_t)=\mathcal{O}(sLKN_t)$ with $K$ being the number of channel uses. Therefore, the total computational complexity for two-stage method is $\mathcal{O}(sL N_r^2)+\mathcal{O}(sLKN_t)= \mathcal{O}(sLKN_t)$. However, for the one-stage OMP method, the computational complexity is $\mathcal{O}(LKNG_tG_r) = \mathcal{O}(s^2LKNN_tN_r)$. It is obvious that the two-stage method has much lower computational complexity compared to the one-stage OMP by $\mathcal{O}(sNN_r)$ times.

For the two-stage method with super resolution, in Stage I, the computational complexity of ADMM per iteration is dominated by the eigenvalue decomposition of $\bX^{t+1}$, i.e., $\mathcal{O}(N_r^3)$. Similarly, for Stage II, each iteration has the computational complexity of $\mathcal{O}(B_{t,2}^3)=\mathcal{O}((K-N_r/N)^3)=\mathcal{O}(K^3)$. Given the number of iteration $T$ and $K\ge N_r$, the total computational complexity of the super resolution method is $\mathcal{O}(TN_r^3)+\mathcal{O}(TK^3)=\mathcal{O}(TK^3)$. In order to compare the complexities of the two-stage method with super resolution and one-stage OMP, we consider a simple example as follows. In particular, if $N_r=N_t$ and $K=\mathcal{O}(N_r)$, the complexity of the proposed two-stage method with super resolution is $\mathcal{O}(s^2LN/T)$ times lower than that of the one-stage OMP.
\section{Conclusion} \label{section conclusion}
{
In this paper, the two-stage method for the mmWave channel estimation was proposed. By sequentially estimating AoAs and AoDs  of large-dimensional antenna arrays, the proposed  two-stage method saved the computational complexity as well as channel use overhead compared to the existing methods.
Theoretically, we analyzed the SRPs of AoA and AoD of the proposed two-stage method. Based on the analyzed SRP,
we designed the resource allocation strategy among two stages to guarantee the accurate AoA and AoD estimation.
In addition, to resolve the issue of quantization error, we extended the proposed two-stage method to a version with super resolution.
The numerical simulations showed that the proposed  two-stage method achieves more accurate channel estimation result than the one-stage method.}

\appendices

\section{Proof of Lemma \ref{lemma SOMP}} \label{appendix5-1}

For an arbitrary random noise matrix $\bN$, the SRP of SOMP has been characterized in \cite{zhang2021successful}. This result is general to be extended to the case in Lemma \ref{lemma SOMP},  where the entries in $\bN$ are \gls{iid} complex Gaussian.

\begin{Theorem}(SRP of SOMP with arbitrary random noise \cite{zhang2021successful}) \label{theorem SOMP random L}
Suppose the signal model provided in Lemma \ref{lemma SOMP}. Given the measurement matrix $\bPhi$ with its MIP constant satisfying $\mu< 1/(2L+1)$ and the cumulative distribution function (CDF) of $\| \bN\|_2$ satisfying
\begin{align}
\text{Pr}(\|  \bN \|_2 \le x) =  F_{N}(x), \label{def prN}
\end{align}
the SRP of SOMP in Algorithm \ref{alg_SOMP} satisfies
 \begin{align}
   \text{Pr}(\cV_S) \ge F_{N}\left(\frac{C_{\text{min}}{{(1-(2L-1)\mu )}} } {2 }\right), \label{SOMPS SRP}
  \end{align}
 where $\cV_{S}$  is the event of successful reconstruction of Algorithm \ref{alg_SOMP}, $C_{\min} = \min\limits _{i\in \Omega}  \lA [\bC]_{i,:} \rA_2$.
\end{Theorem}
According to the results in Theorem \ref{theorem SOMP random L}, the SRP of SOMP is characterized by the CDF of $\| \bN\|_2$. Thus,
in order to extend the result provided in Theorem \ref{theorem SOMP random L} to the case in Lemma \ref{lemma SOMP}, the CDF of $\|\bN\|_2$  is of interest when the entries of $\bN \in \C^{M\times d}$ are \gls{iid} $\cC\cN(0,\sigma^2)$.
Fortunately, according to \cite{TW2, TW1},
the CDF of the largest singular value of $\bN$ converges in distribution to the Tracy-Widom law as $M,d$ tend to  $\infty$,
\begin{align}
\text{Pr}(\| \bN \|_2 \le x)   \approx F_2\left(\frac{ {x^2}/{\sigma^2}-\mu_{M,d}}{\sigma_{M,d}} \right), \label{guanssian dis}
\end{align}
where the function $F_2(\cdot)$ is the CDF of Tracy-Widom law \cite{TW2, TW1}, $\mu_{M,d}  =(M^{1/2}  + d^{1/2}) ^2$, and $
\sigma_{M,d} = (M^{1/2}  + d^{1/2}) (M^{-1/2}  + d^{-1/2}) ^{1/3} $. Finally, after plugging the expression in \eqref{guanssian dis} into \eqref{SOMPS SRP} of Theorem \ref{theorem SOMP random L}, we obtain Lemma \ref{lemma SOMP}, which completes the proof.
\qed

\section{Proof of Proposition \ref{with noise p}} \label{proof prob noise case}
One can write the effective noise as $\tilde{\bN} = \bE + \bN$ where the entries in $\bN$ are \gls{iid} with $\cC\cN(0, \sigma^2)$.
Therefore, we have the following probability bound,
\begin{align}
\Pr\left(  \| {\tilde{\bN}}\|_2 \le x\right) &\overset{(a)}{\le} \Pr\left(\lA {\bE}\rA_2+\lA {\bN}\rA_2 \le x\right) \nonumber\\
&\overset{(b)}{\approx} F_2\left(\frac{{(x-\|  \bE\|_2)^2}/{{\sigma}^2}-\mu_{M,d}}{\sigma_{M,d}}\right),\label{dis with quantize}
\end{align}
where the inequality $(a)$ is due to the triangular inequality, and the approximation $(b)$ holds from \eqref{guanssian dis}.
Then, according to Theorem \ref{theorem SOMP random L},  plugging the expression \eqref{dis with quantize} into \eqref{SOMPS SRP} leads to
\begin{align}
\text{Pr}(\cV_S) \ge F_2\left(\frac{{\left((1-(2L-1)\mu) C_{\text{min}}-2\| \bE\|_2\right)^2}-4\sigma^2 \mu_{M,d}}{4\sigma^2 \sigma_{M,d}}\right), \nonumber
\end{align}
where $C_{\min} = \min\limits _{i\in \Omega}  \lA [\bC]_{i,:} \rA_2$. This concludes the proof.
\qed

\section{Proof of Theorem \ref{AoDs Prob}} \label{appendix5-3}

Plugging  RSB in \eqref{expression Wb2} and TSB in \eqref{expression Fb2}  into \eqref{ob 2nd} gives $\| [\bPhi_2]_{:,j}\|_2 =\sqrt{ {p_2 B_{t,2}}/{N_t}}, ~ j=1,\ldots,G_t$, and $ C_{\min} = \min _{t_l \in \Omega_t}\|  [\bC_2]_{t_l,:} \|_2 =| {{h}_{\min}} | $ with ${{t}_{l}}$ being the index of the $l$th path of $\bA_t$ in $\bar{\bA}_t$ such that $[\bar{\bA}_t]_{:,t_l}=[\bA_t]_{:,l}$, $l=1, \ldots, L$.
Hence, incorporating the latter $C_{\min}$ and $\| [\bPhi_2]_{:, j} \|_2$ into Lemma \ref{lemma SOMP} concludes the proof.
\qed

\bibliographystyle{IEEEtran}

\bibliography{IEEEabrv,Conference_mmWave_CS}

\clearpage

\end{document}